\documentclass[11pt,onecolumn,a4paper]{article}

 \usepackage{amsfonts}
 \usepackage{amsmath}
 \usepackage{amssymb}
 \usepackage{graphicx}
 \usepackage{color}
 \usepackage[english]{babel}
 \usepackage{amsmath,amssymb,booktabs,ctable,graphicx,epsf,epsfig }
\usepackage{psfrag,pstricks}
\usepackage{setspace}

 \headsep
 \baselineskip
\date{}

\begin{document}

\newtheorem{definition}{Definition}[section]
\newtheorem{prop}{Proposition}[section]
\newtheorem{thm}[prop]{Theorem}
\newtheorem{lem}[prop]{Lemma}
\newtheorem{cor}[prop]{Corollary}
\newtheorem{rem}{Remark}[section]
\newtheorem{ex}{Example}[section]
\newtheorem{proof}[prop]{Proof}

\title{On the existence of the stabilizing solution of generalized Riccati equations arising in zero sum stochastic difference games:\\ The time-varying case\footnote{The final version of this paper will be published in Journal of Difference Equations and Applications.}}

\author{Samir Aberkane$^{1,2}$ and Vasile Dragan$^{3,4}$
\\ [1.mm]
\small\it $^1$ Université de Lorraine, CRAN, UMR 7039, Campus Sciences, BP 70239\\[-1mm]
\small\it VandÏuvre-les-Nancy Cedex, 54506, France\\[-1mm]
\small\it $^2$ CNRS, CRAN, UMR 7039, France\\[-1mm]
\small {\tt samir.aberkane@univ-lorraine.fr}\\
\small\it $^3$ Institute of Mathematics "Simion Stoilow"\\[-1mm]
\small\it of the Romanian Academy \\[-1.mm]
\small\it P.O.Box 1-764, RO-014700, Bucharest, Romania   \\ [-1.mm]
\small {\tt Vasile.Dragan@imar.ro}\\
\small\it $^4$ The Academy of the Romanian Scientists, Str. Ilfov, 3, Bucharest, Romania}

\date{}
\maketitle

\begin{abstract}

In this paper, a large class of time-varying Riccati equations arising in stochastic dynamic games is considered. The problem of the existence and uniqueness of some globally defined solution, namely the bounded and stabilizing solution, is investigated. As an application of the obtained existence results, we address in a second step the problem of infinite-horizon zero-sum two players linear quadratic (LQ) dynamic game for a stochastic discrete-time dynamical system subject to both random switching of its coefficients and multiplicative noise. We show that in the solution of such an optimal control problem, a crucial role is played by the unique bounded and stabilizing solution of the considered class of generalized Riccati equations. 

\end{abstract}

\noindent\textbf{Keywords:} Stochastic Riccati equations - stabilizing solution - stochastic control - zero-sum dynamic games.

\section{Introduction}
This paper is devoted to the study of a globally defined solution of a large class of stochastic game-theoretic Riccati difference equations, namely the bounded and stabilizing solution. It is well known that the study of stochastic Riccati equations when compared to their deterministic counterparts offers some very specific challenges that could not be tackled  \textit{mutatis mutandis}. Several important results have been already reported in the literature in this context. We refer here for example to \cite{carte2010,abu,Damm:02} (and the references therein). The results reported in the afore mentioned references are related to Riccati equations with sign definite quadratic terms. The study of global solutions of such \textit{sign definite} nonlinear matrix equations has nowadays reached an advanced maturity state. In this study we are interested by stochastic Riccati equations with \textit{sign indefinite} quadratic terms. Such matrix equations are encountered when dealing with stochastic dynamic games and will be referred to as stochastic game theoretic Riccati equations in the rest of the paper. Unlike their sign definite counterpart, stochastic game theoretic Riccati equations have received less attention and there still exist substantial open problems in the characterization of globally defined solutions of such nonlinear matrix equations (see \cite{mou,Yu:2015} and the reference therein).  In this study, we are interested by a particular globally defined solution, namely the bounded and stabilizing solution. We establish some fundamental properties regarding the solutions of the considered class of matrix difference equations such as uniqueness, monotonicity, comparison as well as existence theorems. Our main contributions along this direction are summarized as follows: 
\begin{itemize}
\item \textbf{Uniqueness:}
Usually the uniqueness of the bounded and stabilizing solution of a Riccati type equation is proved employing some additional properties of this solution such as its maximality or its minimality in a class of admissible solutions. Such arguments cannot be used in the case of Riccati equations without definite sign. For the sign indefinite generalized Riccati equations addressed in this work, we propose an original proof for the uniqueness of the bounded and stabilizing solution. This proof does not rely on other properties of the stabilizing solutions than its boundedness. The main ingredient used here is the exponential stability of adequately defined positive operators on ordered Hilbert spaces. As a consequence of the uniqueness property it is obtained that the bounded and stabilizing solution is a periodic sequence if the coefficients of the Riccati equation under consideration are periodic sequences.
\item \textbf{Existence:} 
We believe that existence conditions of the stabilizing solution for this type of nonlinear matrix equations has not been reported yet in the literature and represents one of our major contributions. The lack of results regarding such a class of Riccati equations is due to the absence of the sign information of the quadratic part of the equation. Our proposed existence conditions are characterized by the non emptiness of an adequately defined matrix set ($\mathcal{A}^\Sigma$) and the stochastic detectability of an adequately defined auxiliary dynamic stochastic system, respectively. We will show that the non emptiness of $\mathcal{A}^\Sigma$ is also a necessary condition (and hence not conservative) while the stochastic detectability condition is a stochastic counterpart of a standard existence condition in deterministic Riccati equations theory.
\end{itemize}
As we already mentioned above, the class of Riccati equations we are dealing with in this work are closely related to stochastic dynamic games. Hence, in the second part of this paper, we will address the problem of infinite-horizon zero-sum two players LQ dynamic game for a stochastic discrete-time dynamical system subject to both random switching of its coefficients and multiplicative noise. We state and solve both the so called \textit{full state-feedback} and \textit{full information feedback} cases. Following a Riccati equation approach, we show that in the solution of such control problems, a crucial role is played by the unique bounded and stabilizing solution of the considered class of generalized Riccati difference equations. We construct optimal strategies in a feedback-loop form by virtue of the bounded and stabilizing solution to the considered Riccati equations verifying specific sign conditions for each admissible strategy. As a byproduct, our proposed solution to this control problem inherits the numerical tractability of the proposed existence conditions for the bounded and stabilizing solution of the game theoretic Riccati equations. We believe that no former results establishing saddle point relations for such games exist in the literature.\\

\noindent This paper is organized as follows: In Section 2 we give the problem setting. Section 3 addresses the uniqueness property of the stabilizing solution of the considered class of Riccati equations. In Section 4 the existence conditions are given. In Section 5 we state and solve the problem of infinite-horizon zero-sum LQ stochastic dynamic game.\\

\noindent\textbf{Notations}. ${\mathfrak N}=\{1,2,...,N\}$ where $N\geq 1$ is a fixed natural number.
$A^T$ stands for the transpose of the matrix $A$ and $Tr[A]$ denotes the trace of a matrix $A$.
The notation $X\geq Y$ ($X>Y$, respectively), where $X$ and $Y$ are symmetric matrices, means that $X-Y$ is positive semi-definite (positive
definite, respectively).
In block matrices, $\star$ indicates symmetric terms: $\left(%
\begin{array}{cc}
 A & B \\
  B^T & C \\
\end{array}%
\right)=\left(%
\begin{array}{cc}
  A & \star \\
  B^T & C \\
\end{array}%
\right)=\left(%
\begin{array}{cc}
  A & B \\
  \star & C \\
\end{array}%
\right)$.
The expression $MN\star$ is equivalent to $MNM^T$ while $M\star$ is equivalent to $MM^T$.
Consider the following space of matrices $\mathcal{M}_{n,m}^N=\mathbb{R}^{n\times m}\times\cdots\times \mathbb{R}^{n\times m}$.
In the case $n=m$ we shall write $\mathcal{M}_{n}^N$ instead of $\mathcal{M}_{n,n}^N$.

We introduce the following convention of notations:
\begin{itemize}
\item If $\mathbb{B}=\left(\begin{array}{ccc}B(1), & \cdots, & B(N)\end{array}\right)\in \mathcal{M}_{n,m}^N$ and $\mathbb{D}=\left(\begin{array}{ccc}D(1), & \cdots, & D(N)\end{array}\right)\in \mathcal{M}_{m,p}^N$ then $\mathbb{C}=\mathbb{B}\mathbb{D}\in\mathcal{M}_{n,p}^N$ where $\mathbb{C}=\left(\begin{array}{ccc}C(1), & \cdots, & C(N)\end{array}\right)$, $C(i)=B(i)D(i)$, $1\leq i\leq N$.
\item $\mathbb{B}^T=\left(\begin{array}{ccc}B^T(1), & \cdots, & B^T(N)\end{array}\right)\in \mathcal{M}_{m,n}^N$.
\item If $\mathbb{A}=\left(\begin{array}{ccc}A(1), & \cdots, & A(N)\end{array}\right)\in \mathcal{M}_{n}^N$ with $\det(A(i))\neq0$, $1\leq i\leq N$, then\\ $\mathbb{A}^{-1}=\left(\begin{array}{ccc}A^{-1}(1), & \cdots, & A^{-1}(N)\end{array}\right)$.
\end{itemize}
As usual, $\mathcal{S}_n\in \mathbb{R}^{n\times n}$ demotes the subspace of symmetric matrices of size $n\times n$ and $\mathcal{S}_n^N=\mathcal{S}_n\times\cdots\times \mathcal{S}_n$.
$\mathcal{S}_n^N$ is a finite dimenional real Hilbert space with respect to the inner product:
\begin{eqnarray}\label{e1a}
\langle {\mathbb X}, {\mathbb Y}\rangle =\sum_{i=1}^N Tr [X(i)Y(i)]
\end{eqnarray}
for all ${\mathbb X}=(X(1), X(2),...,X(N)), {\mathbb Y}=(Y(1),Y(2),...,Y(N))\in \mathcal{S}_n^N$.
Throughout the paper ${\mathbb E}\left[\cdot\right]$ stands for the mathematical expectation and ${\mathbb E}\left[ \cdot | \theta_t=i\right]$ denotes the conditional expectation with respect to the event $\{\theta_t=i\}$.

\section{Problem setting}

\subsection{Model description}

\noindent Consider the following nonlinear difference equation on the space $\mathcal{S}_n^N$:
\begin{equation}\label{GRDE}
 \mathbb{X}(t)=\Pi_{1}(t)[\mathbb{X}(t+1)]+\mathbb{M}(t)-\Big[\Pi_{2}(t)[\mathbb{X}(t+1))]+\mathbb{L}(t)\Big]\Big[\mathbb{R}(t)+
 \Pi_3(t)[\mathbb{X}(t+1)]\Big]^{-1}\star
 \end{equation}
$t\geq 0$, with unknown function $\mathbb{X}(t)=\left(\begin{array}{ccc}X(t,1), & \cdots, & X(t,N)\end{array}\right)$.
Here,\\
$\Pi_k(t)[\mathbb{X}]=\left(\begin{array}{ccc}\Pi_k(t)[\mathbb{X}](1), & \cdots, & \Pi_k(t)[\mathbb{X}](N)\end{array}\right)$, $1\leq k \leq 3$, are defined by:
\begin{equation}\label{PI}
\begin{cases}
\Pi_1(t)[\mathbb{X}](i)=\sum_{j=0}^rA_j^T(t,i)\Xi(t)[\mathbb{X}](i)A_j(t,i)\\
\Pi_2(t)[\mathbb{X}](i)=\sum_{j=0}^rA_j^T(t,i)\Xi(t)[\mathbb{X}](i)B_j(t,i)\\
\Pi_3(t)[\mathbb{X}](i)=\sum_{j=0}^rB_j^T(t,i)\Xi(t)[\mathbb{X}](i)B_j(t,i)\\
\Xi(t)[\mathbb{X}](i)=\overset{N}{\underset{j=1}{\sum}}p_t(i,j)X(j)
\end{cases}
\end{equation}
$1\leq i \leq N$, for all $\mathbb{X}=\left(\begin{array}{ccc}X(1), & \cdots, & X(N)\end{array}\right)\in \mathcal{S}_{n}^N$.
In (\ref{GRDE}) ${\mathbb M}(t)=(M(t,1),...., M(t,N))\in{\cal S}_n^N$, ${\mathbb R}(t)=(R(t,1),...., R(t,N))\in{\cal S}_m^N$, ${\mathbb L}(t)=(L(t,1),..., L(t,N))\in{\cal M}_{n,m}^N$.
Regarding the coefficients of the equation (\ref{GRDE}) we make the assumption:
\begin{itemize}
\item[\textbf{H1)}]
\begin{itemize}
\item[a)] $\{A_j(t,i)\}_{t\geq 0}\subset {\mathbb R}^{n\times n}$ and $\{B_{j}(t,i)\}_{t\geq 0}\subset {\mathbb R}^{n\times m}$, $0\leq j\leq r$,
$\{M(t,i)\}_{t\geq 0} \subset {\cal S}_n$, $\{L(t,i)\}_{t\geq 0}\subset {\mathbb R}^{n\times m}$, $\{R(t,i)\}_{t\geq 0}\subset {\cal S}_m$ for all $ i\in \mathfrak{N}$ are bounded matrix valued sequences.

\item[b)] For each $t\geq 0$  $P_t := (p_t(i,j))_{(i,j)\in{\mathfrak N}\times {\mathfrak N}}$ is a nondegenerate stochastic matrix, i.e.
$p_t(i,j)\geq 0$, $\sum_{k=1}^N p_t(i,k)=1$, $\sum_{k=1}^N p_t(k,j)>0$ for all $i,j \in{\mathfrak N}$.
\end{itemize}
\end{itemize}

\noindent Since the discrete-time backward nonlinear equation (\ref{GRDE}) contains as special cases a large number of discrete time Riccati type equations arising in diverse linear quadratic control problems in both the deterministic and stochastic frameworks, we shall call this type of equations {\bf generalized discrete time Riccati equations (GDTRE)}.

\begin{rem}\label{R2.1}
The GDTRE (\ref{GRDE}) is associated to a quadruple
$$\Sigma=\left(\{{\bf A}(t)\}_{t\geq 0}, \{{\bf B}(t)\}_{t\geq 0}, \{{P_t}\}_{t\geq 0}, \{{\mathbb Q}(t)\}_{t\geq 0} \right)$$
 where ${\bf A}(t)=\left( {\mathbb A}_0(t), {\mathbb A}_1(t),..., {\mathbb A}_r(t)\right)$ with
${\mathbb A}_j(t)=(A_j(t,1),...., A_j(t,N))\in{\cal M}_n^N$,\break
${\bf B}(t)=\left( {\mathbb B}_0(t), {\mathbb B}_1(t),..., {\mathbb B}_r(t)\right)$ with
${\mathbb B}_j(t)=(B_j(t,1),...., B_j(t,N))\in{\cal M}_{n,m}^N$,\\
 ${\mathbb Q}(t)=(Q(t,1),...., Q(t,N))\in{\cal S}_{n+m}^N$ with
\begin{eqnarray}\label{e3a}
Q(t,i)=\left(\begin{array}{cc} M(t,i)&L(t,i)\\L^T(t,i)&R(t,i)\end{array}\right)
\end{eqnarray}
and $P_t\in{\mathbb R}^{N\times N}$ is the stochastic matrix satisfying the assumption {\bf H1)}.
\end{rem}
The GDTRE (\ref{GRDE}) occurs in connection with a LQ control problem described by the controlled system
\begin{equation}\label{equa1}
       x(t+1)=A_0(t, \theta_t)x(t)+B_{0}(t, \theta_t)u(t)+\sum_{k=1}^rw_k(t)\left(A_k(t, \theta_t)x(t)+B_{k}(t, \theta_t)u(t)\right)
\end{equation}
and the quadratic performance criterion
\begin{eqnarray}\label{cost}
\mathcal{J}(t_0,x_0,u)={\mathbb E}\left[\sum_{t=t_0}^\infty\left(\begin{array}{c}x_u(t) \\u(t)\end{array}\right)^T\left(\begin{array}{cc}M(t,\theta_t) & L(t,\theta_t) \\\star & R(t,\theta_t)\end{array}\right)\star\right]
\end{eqnarray}
where $x_u(t)$ is the solution of the initial value problem (IVP) described by the controlled system (\ref{equa1}), $t\geq t_0\geq 0$, and $x(t_0)=x_0$.
In (\ref{equa1}) $\{w_t\}_{t\geq 0}$, $\left(w_t=\left(w_1(t),\cdots,w_r(t)\right)^T\right)$ is a sequence of independent random vectors and the triple $(\{\theta_t\}_{t\geq 0},\{P_t\}_{t\geq 0},\mathfrak{N})$ is a time non-homogeneous Markov chain defined on a given probability space $(\Omega,\mathcal{F},\mathrm{P})$ with the finite states set $\mathfrak{N}=\{1,\cdots, N\}$ and the sequence of transition probability matrices $\{P_t\}_{t\geq 0}$.

\noindent Regarding the processes $\{\theta_t\}_{t\geq 0}$ and $\{w_t\}_{t\geq 0}$, the following assumptions are made:
\begin{itemize}
\item[\textbf{H2)}]  $\{w_t\}_{t\geq 0}$ is a sequence of independent random vectors with the following properties:
\begin{equation*}
{\mathbb E}\left[w(t)\right]=0,\quad {\mathbb E}\left[w(t)w^T(t)\right]=I_r,\quad t\geq 0
\end{equation*}
$I_r$ being the identity matrix of size $r$.

\item[\textbf{H3)}]
\begin{itemize}
\item[a)] For each $t\geq 0$ the $\sigma$-algebra $\mathcal{F}_t$ is independent of the $\sigma$-algebra $\mathcal{G}_t$, where $\mathcal{F}_t=\sigma(w(s);0\leq s\leq t)$ and $\mathcal{G}_t=\sigma(\theta_s;0\leq s\leq t)$.

\item[b)] $\pi_0(i):={\cal P}\{\theta_0=i\}>0$ for all $i\in{\mathfrak N}$.
\end{itemize}
\end{itemize}
The initial probability distributions $\pi_0=\left( \pi_0(1),..., \pi_0(N)\right)$ satisfying the assumption {\textbf H3)} b) will be called {\bf admissible initial probability distributions} of the Markov chain.

\noindent For each $t \geq 0$, we denote ${\mathcal{H}}_t=\mathcal{F}_{t}\vee\mathcal{G}_t$, and for each $t\geq 1$ we denote $\tilde{\mathcal{H}}_t=\mathcal{F}_{t-1}\vee\mathcal{G}_t$ and $\tilde{\mathcal{H}}_0=\sigma(\theta_0)$. Here, ${\mathcal{H}}_t$ is the smallest $\sigma$-algebra containing the $\sigma$-algebras $\mathcal{F}_{t}$ and $\mathcal{G}_{t}$ and $\tilde{\mathcal{H}}_t$ is the smallest $\sigma$-algebra containing the $\sigma$-algebras $\mathcal{F}_{t-1}$ and $\mathcal{G}_{t}$, respectively.\\
\noindent In the following, $\ell_{\tilde{\mathcal{H}}}^2\{t_0,\tau;\mathbb{R}^{m}\}$ stands for the space of all finite sequences $\{y_t\}_{t_0\leq t\leq \tau}$ of $m$-dimensional random vectors with the properties that $\forall$ $t_0\leq t\leq \tau$, $y_t$ is $\tilde{\mathcal{H}}_t$-measurable and $\sum_{t=t_0}^\tau {\mathbb E}\left[|y_t|^2\right]<\infty$.
Also $\ell_{\tilde{\mathcal{H}}}^2\{0,\infty;\mathbb{R}^{m}\}$ is the space of all sequences $\{y_t\}_{t\geq0}$ of $m$-dimensional random vectors with the properties that $\forall t \geq 0$, $y_t$ is $\tilde{\mathcal{H}}_t$-measurable and $\sum_{t=0}^\infty {\mathbb E}\left[|y_t|^2\right]<\infty$.
Both $\ell_{\tilde{\mathcal{H}}}^2\{t_0,\tau;\mathbb{R}^{m}\}$ and  $\ell_{\tilde{\mathcal{H}}}^2\{0,\infty;\mathbb{R}^{m}\}$ are real Hilbert spaces.\\

\noindent The norm induced by the usual inner product on each of these Hilbert spaces are:
\begin{equation*}
\|y\|_{\ell_{\tilde{\mathcal{H}}}^2\{t_0,\tau;\mathbb{R}^{m}\}}=\left(\sum_{t=t_0}^\tau {\mathbb E}\left[|y_t|^2\right]\right)^{\frac{1}{2}},\;\text{for all $y \in \ell_{\tilde{\mathcal{H}}}^2\{t_0,\tau;\mathbb{R}^{m}\}$}
\end{equation*}
and
\begin{equation*}
\|\bar{y}\|_{\ell_{\tilde{\mathcal{H}}}^2\{0,\infty;\mathbb{R}^{m}\}}=\left(\sum_{t=0}^\infty {\mathbb E} \left[|\bar{y}_t|^2\right]\right)^{\frac{1}{2}},\;\text{for all $\bar{y} \in \ell_{\tilde{\mathcal{H}}}^2\{0,\infty;\mathbb{R}^{m}\}$}
\end{equation*}

\noindent We define the class ${\cal{U}}_{ad}(t_0, x_0)$ of admissible controls which consists in all stochastic processes,
$u=\{u(t)\}_{t\geq 0}\in \ell_{\tilde{\mathcal{H}}}^2\{0,\infty;\mathbb{R}^{m}\}$ with the property that $x_u(\cdot)\in \ell_{\tilde{\mathcal{H}}}^2\{0,\infty;\mathbb{R}^{n}\}$.

\begin{rem}\label{R2.2}
If $u\in{\cal U}_{ad}(t_0,x_0)$ then  $\lim\limits_{t\rightarrow\infty}{\mathbb E}[|x_u(t)|^2 | \theta_{t_0}=i]=0$ for all $i\in{\mathfrak N}$.
Indeed  if $u\in{\cal U}_{ad}(t_0,x_0)$ then $\sum_{t=t_0}^\infty {\mathbb E}[|x_u(t)|^2]<\infty$.
Hence, $\lim\limits_{t\rightarrow\infty}{\mathbb E}[|x_u(t)|^2]=0$. On the other hand, according with the assumptions {\bf H1)} b) and  {\bf H3)} b) one may prove inductively that $\pi_{t_0}(i):= {\cal P}\{\theta_{t_0}=i\}>0$ for all $(t_0,i)\in{\mathbb Z}_+\times {\mathfrak N}$.
The conclusion follows now from $0\leq {\mathbb E}[|x_u(t)|^2|\theta_{t_0}=i]\leq \pi_{t_0}^{-1}(i){\mathbb E}[|x_u(t)|^2]$.
\end{rem}

\subsection{The stabilizing solution of GDTRE}

If $\mathbb{F}(t)=\left(\begin{array}{ccc}F(t,1), & \cdots, & F(t,N)\end{array}\right)$ with $F(t,i)\in \mathbb{R}^{m\times n}$, we introduce the following Lyapunov type operator $\mathcal{L}_\mathbb{F}(t):\mathcal{S}_n^N\rightarrow \mathcal{S}_n^N$, defined by:
\begin{align}\label{Lyap}
\mathcal{L}_\mathbb{F}(t)[\mathbb{X}](i)&=\sum_{k=0}^r\sum_{j=1}^N p_t(j,i)\big(A_k(t,j)+B_k(t,j)F(t,j)\big)X(j)\big(A_k(t,j)+B_k(t,j)F(t,j)\big)^T,\; 1\leq i\leq N
\end{align}

\begin{definition}\label{D2.1}
A globally defined solution ${\mathbb X}_s:{\mathbb{Z}}_+\rightarrow{\cal{S}}_n^N$ of (\ref{GRDE}) is named {\em stabilizing solution} if:
\begin{itemize}
\item[a)] \begin{eqnarray}\label{11a}
\underset{t\geq 0}{\inf} |(R(t,i)+\Pi_3(t)[{\mathbb X}_s(t+1)](i))|>0, \forall i\in{\mathfrak N}
 \end{eqnarray}

 \item[b)] the linear difference equation on ${\cal{S}}_n^N$
\begin{equation}\label{e2.7}
{\mathbb{Y}}(t+1)=\mathcal{L}_{\mathbb{F}_s}(t)[\mathbb{Y}(t)]
\end{equation}
is exponentially stable (ES), where $\mathcal{L}_{\mathbb{F}_s}(t)$ is defined as in (\ref{Lyap}) taking $F_s(t,i)$ instead of $F(t,i)$ with:
\begin{equation}\label{e2.8}
F_s(t,i)=-\left[R(t,i)+\Pi_3(t)[\mathbb{X}_s(t+1)](i)\right]^{-1}\left[\Pi_2^T(t)[\mathbb{X}_s(t+1)](i)+L^T(t,i)\right]
\end{equation}
\end{itemize}
\end{definition}

\begin{rem}\label{R2.3}
According with Definition 3.1 and Theorem 3.4 from \cite{carte2010} we obtain that under assumptions {\textbf H1)} b), {\textbf H2)} and {\textbf H3)} a), ${\mathbb X}_s(\cdot)$ is a stabilizing solution of the GDTRE (\ref{GRDE}) if and only if it satisfies (\ref{11a}) and additionally the system of discrete time linear stochastic equations
\begin{eqnarray}\label{11b}
x(t+1)=[A_0(t,\theta_t)+B_0(t,\theta_t)F_s(t,\theta_t)+\sum_{k=1}^r w_k(t)(A_k(t,\theta_t)+B_k(t,\theta_t)F_s(t,\theta_t))]x(t)
\end{eqnarray}
is exponentially stable in mean square with conditioning of first kind (ESMS-CI) i.e. there exist $\beta\geq 1, \alpha\in (0,1)$ such that the solutions of this system satisfy
$${\mathbb E}[|x(t)|^2 |\theta_{t_0}=i]\leq \beta \alpha^{t-t_0}{\mathbb E}[|x(t_0)|^2|\theta_{t_0}=i]$$
for all $t\geq t_0$, $i\in{\mathfrak N}$.
\end{rem}

It is worth pointing out that we do not know a priori neither an initial value nor a boundary value of a stabilizing solution of  the GDTRE (\ref{GRDE}).
That is why, it becomes an important problem to derive necessary and sufficient conditions or at least sufficient conditions that guarantee the existence of a stabilizing solution of this kind of Riccati equations.
In the derivation of some conditions for the existence of a bounded and  stabilizing solution of some GDTRE of type (\ref{GRDE}) an important role is played by the sign of the matrices
\begin{eqnarray}\label{12a}
{\cal R}(t, {\mathbb X}_s(t+1), i):= R(t,i)+\Pi_3(t)[{\mathbb X}_s(t+1)](i).
\end{eqnarray}
Thus the problem of the existence and uniqueness of the bounded and stabilizing solution of a GDTRE of type (\ref{GRDE}) satisfying a sign condition of the form
\begin{eqnarray}\label{12b}
{\cal R}(t,{\mathbb X}_s(t+1), i)\geq \nu I_m>0
\end{eqnarray}
for all $(t,i)\in{\mathbb Z}_+\times {\mathfrak N}$ was studied in Chapter 5 from \cite{carte2010} in the finite dimensional case and in \cite{ungdrmo} in the infinite dimensional case.
The problem of the existence and uniqueness of the bounded and stabilizing solution of a GDTRE of type (\ref{GRDE}) satisfying a sign condition of the form
\begin{eqnarray}\label{12c}
{\cal R}(t,{\mathbb X}_s(t+1), i)\leq -\nu I_m < 0
\end{eqnarray}
for all $(t,i)\in{\mathbb Z}_+\times {\mathfrak N}$ was studied especially with the stochastic $H_{\infty}$ control problem in the so called Bounded Real Lemma (see e.g. \cite{morozan1998, aosr2010} for the case of homogeneous Markov chain and \cite{ma,samirieee,SV} for the case of nonhomogeneous Markov chain).

The goal of the present work is to provide a set of conditions that guarantee the existence and uniqueness of the bounded and stabilizing solution of a GDTRE of type (\ref{GRDE}) satisfying the sign condition:
\begin{eqnarray}\label{12d}
{\rm Inertia}[{\cal R}(t,{\mathbb X}_s(t+1), i)]={\rm Inertia}[{\mathfrak I}]
\end{eqnarray}
where ${\mathfrak I}=diag \{-I_{m_1}, I_{m_2}\}$  and ${\rm Inertia}[\cdot]$ denotes the inertia of a matrix.
A condition of type (\ref{12d}) means that the matrix ${\cal R}(t,{\mathbb X}_s(t+1), i)$ has $m_1$ negative eigenvalues and $m_2$ positive eigenvalues with $m_1+m_2=m$, $m_k\geq 1$ not depend upon $(t,i)\in{\mathbb Z}_+\times {\mathfrak N}$.

Often we shall call GDTRE of type (\ref{GRDE}) with indefinite sign of quadratic part, the Riccati equations satisfying a sign condition of type (\ref{12d}).
First, in the next section, we prove the uniqueness of the bounded and stabilizing solution of a GDTRE of type (\ref{GRDE}).
Further, we shall study the problem of the existence of the bounded and stabilizing solution of this kind of discrete-time Riccati equations satisfying the sign condition (\ref{12d}).
Finally, we shall show how one may use the bounded and stabilizing solution of a GDTRE of type (\ref{GRDE}) to construct an equilibrium strategy in a linear quadratic discrete-time dynamic game.

\section{Uniqueness conditions}

Before stating and proving the main result of this section, we discuss several issues related to the discrete-time perturbed Lyapunov equations.

\noindent First, let us remark that if  $\hat {\cal L}(t):{\cal D}_d^N\to {\cal S}_d^N$ is a discrete-time Lyapunov type operator defined by
\begin{eqnarray}\label{e3.1}
\hat{\cal L}(t)[{\mathbb X}](i)=\sum_{k=0}^r\sum_{j=1}^N p_t(j,i)\hat A_k(t,j)X(j)\hat A^T_k(t,j)
\end{eqnarray}
then its adjoint operator with respect to  the inner product (\ref{e1a}) is given by
\begin{eqnarray}\label{e3.2}
\hat{\cal L}^*(t)[{\mathbb X}](i)=\sum_{k=0}^r\hat A_k^T(t,i)\Xi(t)[{\mathbb X}](i)\hat A_k(t,i)
\end{eqnarray}
for all ${\mathbb X}=(X(1), ..., X(N))\in{\cal S}_d^N$.

\noindent Employing Theorem 2.12 from \cite{carte2010} in the case of the operator (\ref{e3.1}) we obtain:
\begin{cor}\label{C3.1}
If $\{\hat A_k(t,i)\}_{t\geq 0}\subset {\mathbb R}^{d\times d}$ for $0\leq k\leq r, i\in{\mathfrak N}$ are bounded sequences and $p_t(i,j)$ satisfy the assumption {\bf H1)} b), then the following are equivalent:
\begin{itemize}
\item[i)] The discrete-time linear equation on ${\cal S}_d^N$
$$ {\mathbb X}(t+1)=\hat {\cal L}(t)[{\mathbb X}(t)]$$
is exponentially stable.

\item[ii)] The discrete-time backward linear equation
$$ {\mathbb Y}(t)=\hat {\cal L}^*(t)[{\mathbb Y}(t+1)]+{\mathbb I}_d^N$$
has a uniform positive and bounded solution $\tilde {\mathbb Y}(t)=(\tilde Y(t,1),..., \tilde Y(t,N)), t\in{\mathbb Z}_+$ i.e.
$$0< c I_d\leq \tilde Y(t,i)\leq \tilde c I_d$$ for all $(t,i)\in{\mathbb Z}_+\times {\mathfrak N}$
where ${\mathbb I}_d^N=(I_d,..., I_d)\in{\cal S}_d^N$.
\end{itemize}
\end{cor}
Let ${\cal L}_j(t):{\cal S}_n^N \to {\cal S}_n^N, j=1,2$ be the Lyapunov type operators described by
$$ {\cal L}_j(t)[{\mathbb X}](i)=\sum_{k=0}^r \sum_{\ell=1}^N p_t(\ell,i)A_k^j(t,\ell)X(\ell)(A_k^j(t,\ell))^T.$$
Setting $A_{ke}(t,i)=\left(\begin{array}{cc} A_k^1(t,i)&0\\0& A_k^2(t,i)\end{array}\right)$ we define the extended Lyapunov operator
${\cal L}_e(t):{\cal S}_{2n}^N\to {\cal S}_{2n}^N$ by
\begin{eqnarray*}
{\cal L}_e(t)[{\mathbb Z}](i)=\sum_{k=0}^r \sum_{j=1}^N p_t(j,i)A_{ke}(t,j)Z(j)A_{ke}^T(t,j)
\end{eqnarray*}
for all ${\mathbb Z}=(Z(1),..., Z(N))\in{\cal S}_{2n}^N$.
\begin{lem}\label{L3.2}
Assume that $\{A_k^{\ell}(t,i)\}_{t\geq 0}\subset {\mathbb R}^{n\times n}, 0\leq k\leq r, \ell =1,2, i\in{\mathfrak N}$ are bounded sequences and $p_t(i,j)$ satisfy the assumption {\bf H1)} b).
If the discrete-time linear equations on ${\cal S}_n^N$
\begin{eqnarray}\label{e3.3}
{\mathbb X}^{\ell}(t+1) ={\cal L}^{\ell}(t)[{\mathbb X}^{\ell}(t)], \ell =1,2
\end{eqnarray}
are exponentially stable then the discrete-time linear equation on ${\cal S}_{2n}^N$
\begin{eqnarray*}
{\mathbb Z}(t+1)={\cal L}_e(t)[{\mathbb Z}(t)]
\end{eqnarray*}
is exponentially stable too.
\end{lem}
\begin{proof} If the discrete-time linear equations (\ref{e3.3}) are exponentially stable we deduce via Corollary \ref{C3.1} applyied for each of this equations that the backward discrete-time equations
\begin{eqnarray*}
{\mathbb Y}^{\ell}(t)={\cal L}^*_{\ell}(t)[{\mathbb Y}^{\ell}(t+1)]+{\mathbb I}_n^N
\end{eqnarray*}
$\ell=1,2$ have bounded and uniform positive solutions $\tilde {\mathbb Y}^{\ell}(t)=(\tilde Y^{\ell}(t,1),...,\tilde Y^{\ell}(t,N)), t\in{\mathbb Z}_+$.
By direct calculations one obtains that the discrete-time backward equation on ${\cal S}_{2n}^N$
\begin{eqnarray*}
\tilde {\mathbb Z}(t)={\cal L}_e^*(t)[\tilde{\mathbb Z}(t+1)]+{\mathbb I}_{2n}^N
\end{eqnarray*}
has the uniform positive and bounded solution given by  $\tilde {\mathbb Z}(t)=(\tilde Z(t,1),..., \tilde Z(t,N))$ where
$$\tilde Z(t,i)=\left(\begin{array}{cc} \tilde Y^1(t,i)& 0\\ 0& \tilde Y^2(t,i)\end{array}\right).$$
Applying the implication $(ii)\to (i)$ from Corollary \ref{C3.1} in the case of the equation
$${\mathbb Z}(t+1)={\cal L}_e(t)[{\mathbb Z}(t)]$$
we obtain the desired conclusion. Thus the proof ends. \hfill$\blacksquare$
\end{proof}

We are now in position to prove the main result of this section.

\begin{thm}\label {T3.3}
Assume that the assumption {\bf H1)} is fulfilled. Under this condition the following hold:
\begin{itemize}
\item[i)] the (GDTRE) (\ref{GRDE}) has at most one bounded on ${\mathbb Z}_+$ and stabilizing solution;

\item[ii)] if there exists an integer ${\mathfrak p}\geq 1$ such that $A_k(t+{\mathfrak p}, i)=A_k(t,i)$, $B_k(t+{\mathfrak p}, i)=B_k(t,i)$ for $0\leq k\leq r$, $M(t+{\mathfrak p}, i)=M(t,i)$, $L(t+{\mathfrak p}, i)=L(t,i)$, $R(t+{\mathfrak p}, i)=R(t,i)$, $p_{t+{\mathfrak p}}(i,j)=p_t(i,j)$ for all $t\in{\mathbb Z}_+, i,j\in{\mathfrak N}$, then the unique bounded and stabilizing solution of (\ref{GRDE}) if any is a periodic sequence of period ${\mathfrak p}$.
\end{itemize}
\end{thm}

\begin{proof}
\begin{itemize}
\item[i)]  Let us assume by contrary that the GDTRE (\ref{GRDE}) has two bounded and stabilizing solutions $\{{\mathbb X}_s^{\ell}(t)\}_{t\in{\mathbb Z}_+}, \ell=1,2$.
Let $F_s^{\ell}(t,i)$ be the corresponding stabilizing feedback gain associated to the solution ${\mathbb X}_s^{\ell}(\cdot)$ via (\ref{e2.8}).
By direct calculations one obtains that the equation (\ref{GRDE}) verified by ${\mathbb X}_s^{\ell}(\cdot)$ may be written in the form:
\begin{align}\label{e3.4}
X_s^{\ell}(t,i)&=\sum_{k=0}^r [A_k(t,i)+B_k(t,i)F_s^1(t,i)]^T\Xi(t)[{\mathbb X}_s^{\ell}(t+1)](i)[A_k(t,i)+B_k(t,i)F_s^2(t,i)]\notag\\
&+\left(I_n \quad (F_s^1(t,i))^T\right) Q(t,i)\left(I_n \quad (F_s^2(t,i))^T\right)^T, \ell=1,2
\end{align}
where $Q(t,i)$ were introduced by (\ref{e3a}).
Let ${\mathbb Y}(t)=(Y(t,1),...,Y(t,N))$ be defined by
$$Y(t,i)=X_s^1(t,i)-X_s^2(t,i), \;\; (t,i)\in{\mathbb Z}_+\times {\mathfrak N}.$$
From (\ref{e3.4}) we obtain that ${\mathbb Y}(\cdot)$ satisfies the discrete-time backward linear equation on ${\cal S}_n^N$
\begin{eqnarray}\label{e3.5bis}
Y(t,i)=\sum_{k=0}^r [A_k(t,i)+B_k(t,i)F_s^1(t,i)]^T\Xi(t)[{\mathbb Y}_s^{\ell}(t+1)](i)[A_k(t,i)+B_k(t,i)F_s^2(t,i)]
\end{eqnarray}
for all $(t,i)\in{\mathbb Z}_+\times {\mathfrak N}$.
We set
\begin{eqnarray}\label{e3.6}
\hat A_k(t,i)=\left(\begin{array}{cc} A_k(t,i)+B_k(t,i)F_s^1(t,i)&0\\0&  A_k(t,i)+B_k(t,i)F_s^2(t,i)\end{array}\right)\in{\mathbb R}^{2n\times 2n},\quad 0\leq k\leq r.
 \end{eqnarray}
Based on this choice of the matrices $\hat A_k(t,i)$ we construct  the Lyapunov type operator $\hat{\cal L}_e(t)$ and its adjoint operator $\hat {\cal L}^*_e(t)$ defined as in (\ref{e3.1}) and (\ref{e3.2}), respectively.
On the other hand, let ${\cal L}_{{\mathbb F}_s^{\ell}}(t):{\cal S}_n^N\to {\cal S}_n^N$ be the linear operators defined as in (\ref{Lyap}) for ${\mathbb F}(t)$ replaced by ${\mathbb F}_s^{\ell}(t)=(F_s^{\ell}(t,1),..., F_s^{\ell}(t,N))$.
Since ${\mathbb X}_s^{\ell}(\cdot)$ is a stabilizing solution of (\ref{GRDE}) it follows that the discrete-time linear equation
\begin{eqnarray*}
{\mathbb X}(t+1)={\cal L}_{{\mathbb F}_s^{\ell}}(t)[{\mathbb X}(t)]
\end{eqnarray*}
is exponentially stable.
Applying Lemma \ref{L3.2} we deduce that the discrete-time equation on ${\cal S}_{2n}^N$
$$ {\mathbb Z}(t+1)=\hat{\cal L}_e(t)[{\mathbb Z}(t)]$$
is exponentially stable.
Let $\tilde {\mathbb Z}(t)=(\tilde Z(t,1),..., \tilde Z(t,))$ be defined by
$$ \tilde Z(t,i)= \left(\begin{array}{cc} 0&X_s^1(t,i)-X_s^2(t,i)\\X_s^1(t,i)-X_s^2(t,i)&0\end{array}\right) \in{\cal S}_{2n}.$$
By direct calculation, involving (\ref{e3.2}), (\ref{e3.5bis}), (\ref{e3.6}) we obtain that $\{\tilde{\mathbb Z}(t)\}_{t\geq 0}$ is a bounded solution of a discrete-time backward equation
\begin{eqnarray}\label{e3.7}
 \tilde {\mathbb Z}(t)=\hat {\cal L}_e^*(t)[\tilde{\mathbb Z}(t+1)]
\end{eqnarray}
Applying Theorem 2.5  from \cite{carte2010} to the special case of the equation (\ref{e3.7}) we deduce that this equation has a unique bounded solution namely  $\tilde{\mathbb Z}(t)=0\in{\cal S}_{2n}^N$, $t\in{\mathbb Z}_+$.
This allows us to conclude that $X_s^1(t,i)=X_s^2(t,i)$ for all $(t,i)\in{\mathbb Z}_+\times {\mathfrak N}$. Thus we have shown that part $i)$ from the statement holds.

\item [ii)] Assume that the GDTRE (\ref{GRDE}) has a bounded and stabilizing solution ${\mathbb X}_s(\cdot)=(X_s(\cdot, 1),..., X_s(\cdot, N))$.
We set $\hat{\mathbb X}(t)=(\hat X(t,1),..., \hat X(t, N))$ defined by
$\hat X(t,i)=X_s(t+{\mathfrak p},i)$, $t\in{\mathbb Z}_+$.
By direct calculation involving the periodicity of the coefficients one shows that $\{\hat{\mathbb X}(t)\}_{t\geq 0}$ is a bounded solution of (\ref{GRDE}).
Let ${\cal L}_{\hat{\mathbb F}}(t)$ be the linear operator defined as in (\ref{Lyap}) for  ${\mathbb F}(t)$ replaced by $\hat{\mathbb F}(t)=(\hat F(t,1),..., \hat F(t,N)), \; \hat F(t,i)$ being computed as in (\ref{e2.8}) for $X_s(t,i)$ replaced by $\hat X(t,i)$.
Employing the periodicity property of the coefficients of (\ref{GRDE}) one obtains that
\begin{eqnarray}\label{e3.8}
{\cal L}_{\hat {\mathbb F}}(t)={\cal L}_{{\mathbb F}_s}(t+{\mathfrak p}), \quad t\in{\mathbb Z}_+.
\end{eqnarray}
Let ${\mathbb T}_s(t,t_0)={\cal L}_{{\mathbb F}_s}(t-1){\cal L}_{{\mathbb F}_s}(t-2)\cdots{\cal L}_{{\mathbb F}_s}(t_0)$  if $t>t_0\geq 0$, ${\mathbb T}_s(t_0,t_0)=\mathbb{I}_{{\cal S}_n^N}$ and
$\hat{\mathbb T}(t,t_0)= {\cal L}_{\hat {\mathbb F}}(t-1){\cal L}_{\hat {\mathbb F}}(t-2)... {\cal L}_{\hat {\mathbb F}}(t_0)$ if $t>t_0\geq 0, \hat {\mathbb T}(t_0,t_0)= \mathbb{I}_{{\cal S}_n^N}$ be the linear evolution operators  defined by ${\cal L}_{{\mathbb F}_s}(\cdot)$ and ${\cal L}_{\hat{\mathbb F}}(\cdot)$, respectively, where $\mathbb{I}_{{\cal S}_n^N}$is the identity operator on ${\cal S}_n^N$.

\noindent Based on (\ref{e3.8}) we get
\begin{eqnarray}\label{e3.9}
\hat {\mathbb T}(t,t_0)={\mathbb T}_s(t+{\mathfrak p}, t_0+{\mathfrak p}), \quad \forall\;\;  t\geq t_0\geq 0.
\end{eqnarray}
From the definition of the stabilizing solution ${\mathbb X}_s(\cdot)$ we deduce that $\|{\mathbb T}_s(t,t_0)\|\leq \beta \alpha^{t-t_0}$ for all $t\geq t_0\geq 0$, where $\beta\geq 1$, $\alpha\in(0,1)$ are constants.
Further (\ref{e3.9}) yields
$\|\hat {\mathbb T}(t,t_0)\|=\|{\mathbb T}_s(t+{\mathfrak p}, t_0+{\mathfrak p})\|\leq \beta \alpha^{t-t_0}$. hence the discrete-time linear equation on ${\cal S}_n^N$
$${\mathbb X}(t+1)={\cal L}_{\hat{\mathbb F}}(t)[{\mathbb X}(t)]$$
is exponentially stable, which means that $\hat {\mathbb X}(\cdot)$ is also a bounded and stabilizing solution of the GDTRE (\ref{GRDE}).
From the uniqueness of the bounded and stabilizing solution of (\ref{GRDE}) we conclude that $X_s(t,i)\hat X(t,i)=X_s(t+{\mathfrak p},i)$ for all $(t,i)\in{\mathbb Z}_+\times {\mathfrak N}$. Thus the proof ends.\hfill$\blacksquare$
\end{itemize}
\end{proof}
\begin{rem}\label{R3.1}
In the case when the bounded and stabilizing solution ${\mathbb X}_s(\cdot)$ of a GDTRE of type (\ref{GRDE}) satisfies a sign condition of type (\ref{12b}) or (\ref{12c})  one shows that in this case ${\mathbb X}_s(\cdot)$ is the bounded and maximal solution (bounded and minimal solution), respectively of the considered GDTRE, which guarantee the uniqueness of the bounded  and stabilizing solution.

In the case when the bounded and stabilizing solution of a GDTRE of type (\ref{GRDE}) satisfies a sign condition of type (\ref{12d}) we cannot say if this solution is maximal or minimal and, therefore, in this case we cannot obtain the uniqueness of the bounded and stabilizing solution based on the uniqueness of the maximal or minimal solution of this kind of GDTRE.
The proof of Theorem \ref{T3.3} from above is not based on the property of the sign of the matrices defined in (\ref{12a}).
\end{rem}

\section{Existence conditions}
In this section we deal with the problem of the existence of the bounded and stabilizing solution of a GDTRE (\ref{GRDE}) satisfying a sign condition of type (\ref{12d}).
We shall provide some conditions which guarantee the existence of the bounded and stabilizing solution of this kind of Riccati equations with the desired sign condition.
\subsection{Preliminary remarks}
Consider the following partitions of the coefficients of (\ref{GRDE}):
\begin{eqnarray}
\label{partition1}
B_j(t,i)=\left(
           \begin{array}{cc}
             B_{j1}(t,i) & B_{j2}(t,i) \\
           \end{array}
         \right),\;B_{jk}(t,i)\in{\mathbb{R}}^{n\times m_k}, \;0\leq j\leq r,\\
L(t,i)=\left( L_1(t,i)\quad L_2(t,i)\right), L_k(t,i)\in{\mathbb R}^{n\times m_k},k=1,2\nonumber
\end{eqnarray}
and:
\begin{equation}
\label{partition2}
R(t,i)=\left(\begin{array}{cc}R_{11}(t,i) & R_{12}(t,i) \\\star & R_{22}(t,i)\end{array}\right),\;R_{lj}(t,i)\in{\mathbb{R}}^{m_l\times m_j}, \; l,\;j=1,2
\end{equation}
\noindent Correspondingly we obtain the partitions:
\begin{equation}
\begin{cases}
\Pi_2(t)[\mathbb{X}](i)=\left(\begin{array}{cc}\Pi_{21}(t)[\mathbb{X}](i) & \Pi_{22}(t)[\mathbb{X}](i)\end{array}\right)\\
\Pi_3(t)[\mathbb{X}](i)=\left(\begin{array}{cc}\Pi_{311}(t)[\mathbb{X}](i) & \Pi_{312}(t)[\mathbb{X}](i) \\\star & \Pi_{322}(t)[\mathbb{X}](i)\end{array}\right)
\end{cases}
\end{equation}
with
\begin{equation*}
\begin{cases}
\Pi_{2k}(t)[\mathbb{X}](i)=\sum_{j=1}^rA_j^T(t,i)\Xi(t)[\mathbb{X}](i)B_{jk}(t,i)\\
\Pi_{3lk}(t)[\mathbb{X}](i)=\sum_{j=1}^rB_{jl}^T(t,i)\Xi(t)[\mathbb{X}](i)B_{jk}(t,i)
\end{cases}; k,l=1,2
\end{equation*}

\noindent In the rest of the paper, the following assumption regarding the weight matrices $M(t,i),\; R(t,i)$ and $L(t,i)$ is made:

\begin{itemize}
\item[\textbf{H4)}] For each $(t,i) \in \mathbb{Z}_+\times \mathfrak{N}$:
\begin{equation}\label{sign1}
R_{22}(t,i)\geq \rho_2 I_{m_2}
\end{equation}
\begin{equation}\label{sign2}
M(t,i)-L_2(t,i)R_{22}^{-1}(t,i)L_2^T(t,i)\geq 0
\end{equation}
\begin{equation}\label{sign3}
R_{11}(t,i)-R_{12}(t,i)R_{22}^{-1}(t,i)R_{12}^T(t,i)\leq -\rho_1I_{m_1}
\end{equation}
with $\rho_j>0$, $j=1,2$  given constant scalars.
\end{itemize}

\noindent In this work, we are interested by solutions $\mathbb{X}(\cdot):{\cal{I}}\subset{\mathbb{Z}_+} \rightarrow{\cal{S}}_n^N$ of (\ref{GRDE}) satisfying the following sign conditions:
\begin{align}\label{e2.3}
&{\cal R}_{22}^{\sharp}(t,{\mathbb X}(t+1),i)=R_{11}(t,i)+\Pi_{311}(t)[\mathbb{X}(t+1)](i)-\Big[R_{12}(t,i)+\Pi_{312}(t)[\mathbb{X}(t+1)](i)\Big]
\notag\\&\times\Big[R_{22}(t,i)+\Pi_{322}(t)[\mathbb{X}(t+1)](i)\Big]^{-1}\star\leq -\delta_1I_{m_1}
\end{align}
\begin{equation}\label{e2.4}
{\cal R}_{22}(t,{\mathbb X}(t+1),i)=R_{22}(t,i)+\Pi_{322}(t)[\mathbb{X}(t+1)](i)\geq \delta_2I_{m_2}
\end{equation}
for all $t\in{\cal{I}}$, $1\leq i\leq N$, $\delta_k>0$, $k=1,2$, being constants.

\begin{rem}\label{R4.1a}
Let $\left(\begin{array}{cc}{\cal R}_{11}(t,{\mathbb X}(t+1),i)& {\cal R}_{12}(t,{\mathbb X}(t+1),i)\\
{\cal R}_{12}^T(t,{\mathbb X}(t+1),i)& {\cal R}_{22}(t,{\mathbb X}(t+1),i)
\end{array}\right)$ be the partition of the matrix\break ${\cal R}(t,{\mathbb X}(t+1),i)$ compatible with the partition (\ref{partition1}), ${\cal R}(t,{\mathbb X}(t+1),i)$ being computed as in (\ref{12a}) with ${\mathbb X}_s(t+1)$ being replaced by ${\mathbb X}(t+1)$. So, the left part of (\ref{e2.4}) is the $(2,2)$-block of the matrix ${\cal R}(t,{\mathbb X}(t+1),i)$, while the left part of (\ref{e2.3}) is the Schur complement of the $(2,2)$-block of the same matrix.
\end{rem}
\noindent In the sequel, we shall denote ${\cal D}^{\mathfrak R}$ the set of solutions $\{{\mathbb X}(t)\}_{t\in{\cal I}}$ of the GDTRE (\ref{GRDE}) satisfying the conditions (\ref{e2.3}) and (\ref{e2.4}) and we call them admissible solutions. Employing Lemma 2 in Chapter 4 from \cite{Halanay:93}, some useful  properties of admissible solutions can be derived:
\begin{cor}\label{C4.1}
Assume that the Assumptions {\bf H1)} a) and {\bf H4)} are fulfilled.\\
Let $\tilde {\mathbb X}(t)=(\tilde X(t,1),..., \tilde X(t,N)), t\in{\mathbb Z}_+$ be a bounded solution of (\ref{GRDE}).
 If $\tilde X(\cdot)\in {\cal D}^{\mathfrak R}$ then the matrices ${\cal R}(t,\tilde{\mathbb X}(t+1), i)$ satisfy a sign condition of type (\ref{12d}).
 If $\tilde {\mathbb X}(\cdot)$ and the coefficients of (\ref{GRDE}) are periodic sequences of period ${\mathfrak p}$ and additionally $\tilde X(t,i)\geq 0$ for all $(t,i)\in{\mathbb Z}_+\times {\mathfrak N}$, then the following are equivalent:
\begin{itemize}
 \item[i)] $\tilde {\mathbb X}(\cdot)\in{\cal D}^{\mathfrak R}$
\item[ii)] the matrices ${\cal R}(t,\tilde{\mathbb X}(t+1),i), (t,i)\in{\mathbb Z}_+\times{\mathfrak N}$ satisfy a sign condition of type (\ref{12d}).
\end{itemize}
\end{cor}
According to the result stated in Corollary \ref{C4.1}, to derive conditions that guarantee the existence of a bounded and stabilizing solution of (\ref{GRDE}) satisfying a sign condition of type (\ref{12d}) we shall provide conditions that allows us to obtain the existence of a bounded and stabilizing solution of (\ref{GRDE}) satisfying sign conditions of type (\ref{e2.3})-(\ref{e2.4}).
We shall see that the conditions of type (\ref{e2.3})-(\ref{e2.4}) are more adapted to the definition of some admissible strategies in a zero sum discrete-time LQ dynamic game.

\begin{rem}\label{R4.2abis}
In the sequel, we will also give conditions for the existence of a bounded and stabilizing solution of (\ref{GRDE}) satisfying sign conditions of type:
\begin{equation}\label{e2.3'}
{\cal R}_{11}(t,{\mathbb X}(t+1),i)=R_{11}(t,i)+\Pi_{311}(t)[\mathbb{X}(t+1)](i)\leq-\delta_3I_{m_1}
\end{equation}
together with  (\ref{e2.4}), for all $t\in{\cal{I}}$, $1\leq i\leq N$, $\delta_3>0$ being a constant. Such solution is involved in the definition of another admissible strategy in a zero sum discrete-time LQ dynamic game. Such a strategy appears for example in \cite{mou} where further details could be found for the interested reader.
\end{rem}

\begin{rem}\label{R4.2a}
If $\mathbb{X}(\cdot):\mathcal{I}\rightarrow \mathcal{S}_n^N$ is an admissible solution of GDTRE (\ref{GRDE}) then we have the following factorization:
\begin{equation}\label{eq.34.rev}
R(t,i)+\Pi_3(t)[\mathbb{X}(t+1)](i)=\left(\begin{array}{cc}V_{11}(t)[\mathbb{X}(t+1)](i) & \mathbf{0} \\V_{21}(t)[\mathbb{X}(t+1)](i) & V_{22}(t)[\mathbb{X}(t+1)](i)\end{array}\right)^T\left(\begin{array}{cc}-{ I}_{m_1} & \mathbf{0} \\ \mathbf{0} & {I}_{m_2}\end{array}\right)\star
\end{equation}
where $V_{kk}(t)[\mathbb{X}(t+1)](i)\geq c_k { I}_{m_k}$, $k=1,2$, $i\in \mathfrak{N}$ and $t\in \mathcal{I}$. In order to obtain (\ref{eq.34.rev}) we may take
$$V_{11}(t)[{\mathbb X}(t+1)](i) =(-{\cal R}_{22}^{\sharp}(t, {\mathbb X}(t+1),i))^{\frac{1}{2}}$$
$$V_{21}(t)[{\mathbb X}(t+1)](i) =({\cal R}_{22}(t, {\mathbb X}(t+1),i))^{-\frac{1}{2}}{\cal R}_{12}^T(t,{\mathbb X}(t+1),i)$$
$$V_{22}(t)[{\mathbb X}(t+1)](i) =({\cal R}_{22}(t, {\mathbb X}(t+1),i))^{\frac{1}{2}}.$$
This factorization, adapted to the indefinite sign of $R(t,i)+\Pi_3(t)[{\mathbb X}(t+1)](i)$,
 will play a crucial role in the derivation of the main results.
\end{rem}

\noindent We will now define and characterize the set $\mathcal{A}^\Sigma$ that will play a key role in the proof of our main results. To this end, and according with the partition (\ref{partition1}), we define $u(t)=\left(\begin{array}{cc}u_1^T(t) & u_2^T(t)\end{array}\right)^T$ and we set formally $u_2(t)\equiv u_2^{\mathbb{KW}}(t)=K(t,\theta_t)x(t)+W(t,\theta_t)u_1(t)$.
Hence (\ref{equa1}) and (\ref{cost}) is rewritten as follows:
\begin{equation}\label{equa3}
       x(t+1)=A_{0\mathbb{K}}(t,\theta_t)x(t)+B_{0\mathbb{W}}(t, \theta_t)u_1(t)+\sum_{k=1}^rw_k(t)\left(A_{k\mathbb{K}}(t, \theta_t)x(t)+B_{k\mathbb{W}}(t, \theta_t)u_1(t)\right)\end{equation}
\begin{equation}
\label{cost_closed}
\mathcal{J}_{\mathbb{KW}}(t_0,x_0,u_1)={\mathbb E}\left[\sum_{t=t_0}^\infty\left(\begin{array}{c}x_{u_1}(t) \\u_1(t)\end{array}\right)^T\left(\begin{array}{cc}M_\mathbb{K}(t,\theta_t) & L_{\mathbb{KW}}(t,\theta_t) \\\star & R_\mathbb{W}(t,\theta_t)\end{array}\right)\star\right]
\end{equation}
where $x_{u_1}(t)$ is the solution of (\ref{equa3}) corresponding to $u_1(t)$ and:
\begin{equation}\label{notation}
\begin{cases}
A_{k\mathbb{K}}(t,i)=A_k(t,i)+B_{k2}(t,i)K(t,i)\\
B_{k\mathbb{W}}(t,i)=B_{k1}(t, i)+B_{k2}(t,i)W(t,i)\\
M_{\mathbb{K}}(t,i)=M(t,i)+L_2(t,i)K(t,i)+K^T(t,i)L_2^T(t,i)+K^T(t,i)R_{22}(t,i)K(t,i)\\
L_{\mathbb{KW}}(t,i)=L_1(t,i)+K^T(t,i)R_{12}^T(t,i)+\left(L_2(t,i)+K^T(t,i)R_{22}(t,i)\right)W(t,i)\\
R_{\mathbb{W}}(t,i)=\left(\begin{array}{c}I_{m_1} \\W(t,i)\end{array}\right)^TR(t,i)\left(\begin{array}{c}I_{m_1} \\W(t,i)\end{array}\right).
\end{cases}
\end{equation}
To the above system (\ref{equa3}) and the corresponding quadratic functional (\ref{cost_closed}) we associate the following Riccati-type difference equation of type (\ref{GRDE}):
\begin{align}\label{Ric_closed}
{X}(t,i)&=\Pi_\mathbb{K}(t)[\mathbb{X}(t+1)](i)+M_\mathbb{K}(t,i)-\big(\Pi_{\mathbb{K}\mathbb{W}}(t)[\mathbb{X}(t+1)](i)+L_\mathbb{KW}(t,i)\big)\times\notag\\
&\times\big(R_\mathbb{W}(t,i)+\Pi_\mathbb{W}(t)[\mathbb{X}(t+1)](i)\big)^{-1}\star
\end{align}
where:
\begin{equation}
\begin{cases}
\Pi_\mathbb{K}(t)[\mathbb{X}](i)=\sum_{j=0}^rA_{j\mathbb{K}}^T(t,i)\Xi(t)[\mathbb{X}](i)A_{j\mathbb{K}}(t,i)\\
\Pi_{\mathbb{K}\mathbb{W}}(t)[\mathbb{X}](i)=\sum_{j=0}^rA_{j\mathbb{K}}^T(t,i)\Xi(t)[\mathbb{X}](i)B_{j\mathbb{W}}(t,i)\\
\Pi_\mathbb{\mathbb{W}}(t)[\mathbb{X}](i)=\sum_{j=0}^rB_{j\mathbb{W}}^T(t,i)\Xi(t)[\mathbb{X}](i)B_{j\mathbb{W}}(t,i)
\end{cases}
\end{equation}
for all $\mathbb{X}\in \mathcal{S}_n^N$.
Since the Riccati equation  (\ref{Ric_closed}) is of type (\ref{GRDE}) it follows that the notion of stabilizing solution of (\ref{Ric_closed}) may be introduced as in Definition \ref{D2.1}.
Furthermore, according to Theorem \ref{T3.3} the GDTRE (\ref{Ric_closed}) may have at most one bounded and stabilizing solution.\\
In the following, to a GDTRE (\ref{GRDE}) defined by a given quadruple of sequences
$$\Sigma=(\{{\bf A}(t)\}_{t\geq 0}, \{{\bf B}(t)\}_{t\geq 0}, \{P_t\}_{t\geq 0}, \{{\mathbb Q}(t)\}_{t\geq 0}),$$
we associate the set ${\cal{A}}^{\Sigma}$, which consists of all pairs of feedback gains $(\mathbb{K}(\cdot), \mathbb{W}(\cdot))$ where $t\rightarrow\mathbb{K}(t)=\left(\begin{array}{ccc}K(t,1), & \cdots, & K(t,N)\end{array}\right):\mathbb{Z}_+\rightarrow \mathcal{M}_{m_2,n}^N$, $t\rightarrow\mathbb{W}(t)=\left(\begin{array}{ccc}W(t,1), & \cdots, & W(t,N)\end{array}\right):\mathbb{Z}_+\rightarrow \mathcal{M}_{m_2,m_1}^N$ are
bounded matrix valued sequences having the properties:
\begin{itemize}
\item[i)]  The zero solution of the stochastic linear system:
\begin{align}
\label{sys_closed}
x(t+1)&=A_{0\mathbb{K}}(t,\theta_t)x(t)+\sum_{k=1}^rw_k(t)A_{k\mathbb{K}}(t,\theta_t)x(t)
\end{align}
is ESMS-CI (see Definition 3.1 from \cite{carte2010} for details).
\item[ii)] The corresponding GDTRE (\ref{Ric_closed}) has a unique bounded and stabilizing solution $\tilde{\mathbb{X}}_{\mathbb{K}\mathbb{W}}(\cdot)$ satisfying the sign condition:
\begin{equation}
\label{sign_closed}
R_\mathbb{W}(t,i)+\Pi_\mathbb{W}(t)[\tilde{\mathbb{X}}_{\mathbb{K}\mathbb{W}}(t+1)](i)\leq -\xi \mathbb{I}
\end{equation}
for some positive scalar $\xi$ (which may depend upon $\left(\mathbb{K}(\cdot),\mathbb{W}(\cdot)\right)$), $(t,i)\in \mathbb{Z}_+\times \mathfrak{N}$.
\end{itemize}
\begin{rem}\label{R4.1}
\begin{itemize}
\item [i)] In order to obtain necessary and sufficient conditions that allow us to decide if the set $\mathcal{A}^{\Sigma}$ is empty or not, one can apply Theorem 5.6 in \cite{carte2010} to the Riccati difference equation:
\begin{align}\label{Ric_closed_modif}
{Y}(t,i)&=\Pi_\mathbb{K}(t)[\mathbb{Y}(t+1)](i)-M_\mathbb{K}(t,i)-\big(\Pi_{\mathbb{K}\mathbb{W}}(t)[\mathbb{Y}(t+1)](i)-L_\mathbb{W}(t,i)\big)\notag\\
&\times\big(-R_\mathbb{W}(t,i)+\Pi_\mathbb{W}(t)[\mathbb{Y}(t+1)](i)\big)^{-1}\star
\end{align}
obtained from (\ref{Ric_closed}) by taking $Y(t,i)=-X(t,i)$, $(t,i)\in \mathbb{Z}_+\times \mathfrak{N}$ and noticing that $\tilde{\mathbb{X}}(t)$ is the bounded and stabilizing solution of (\ref{Ric_closed}) satisfying the sign condition (\ref{sign_closed}) if and only if $\tilde{\mathbb{Y}}(t)$ is the bounded and stabilizing solution of (\ref{Ric_closed_modif}) satisfying the sign condition:
\begin{equation}
\label{sign_closed_modif}
\tilde{R}_\mathbb{W}(t,i)+\Pi_\mathbb{W}(t)[\tilde{\mathbb{Y}}(t+1)](i)\geq \nu \mathbb{I}
\end{equation}
for some positive scalar $\nu$ and $\tilde{R}_\mathbb{W}(t,i)=-R_\mathbb{W}(t,i)$,  $(t,i)\in \mathbb{Z}_+\times \mathfrak{N}$. Such a necessary and sufficient condition is formulated as a convex optimization problem via an LMIs (Linear Matrix Inequalities) setting.
\item [ii)] It follows from the exponential stability in the mean square sense with conditioning of system (\ref{sys_closed}) and sign conditions (\ref{sign1}) and (\ref{sign2}) one can show that if $\left(\mathbb{K}(\cdot), \mathbb{W}(\cdot)\right)\in \mathcal{A}^{\Sigma}$ then the bounded and stabilizing solution of (\ref{Ric_closed}) is positive semi-definite.
\end{itemize}
\end{rem}
Among the subsets of ${\cal A}^{\Sigma}$ we distinguish the subset ${\cal A}_{0}^{\Sigma}$ that consists of all feedback gains $({\mathbb K}(\cdot),{\mathbb W}(\cdot))\in{\cal A}^{\Sigma}$ such that $W(t,i)=0$ for all $(t,i)\in{\mathbb Z}_+\times{\mathfrak N}$.
The bounded and stabilizing solution of the GDTRE (\ref{Ric_closed}) with $W(t,i)=0$ will be denoted by $\tilde{\mathbb{X}}_{\mathbb{K}}(t)=(\tilde X_{\mathbb{K}}(t,1),\cdots, \tilde X_{\mathbb{K}}(t, N))$.
The sign condition (\ref{sign_closed}) reduces in this case to:
\begin{equation}\label{sign_closed'}
R_{11}(t,i)+\sum\limits_{k=0}^rB_{k1}^T(t,i){\tilde{X}}_{\mathbb{K}}(t,i)B_{k1}(t,i)\leq -\xi I_{m_1}.
\end{equation}
\noindent Note that Remark \ref{R4.1} could be specialized in a straightforward way to the set ${\cal A}_{0}^{\Sigma}$ and the corresponding stabilizing solution $\tilde{\mathbb{X}}_{\mathbb{K}}(t)$ of the GDTRE (\ref{Ric_closed}). Hence by combining the sign condition (\ref{sign_closed'}) with (\ref{sign1}) and (\ref{sign2}) one shows easily that $\tilde{\mathbb{X}}_{\mathbb{K}}(t)$ verifies the sign conditions (\ref{e2.3}) and (\ref{e2.4}).

\subsection{Main results}
Before introducing the main existence conditions, let us first start by stating and proving a comparison theorem for the considered class of generalized Riccati equations. Such a result shows the monotonicity of the admissible solutions (in the sense of conditions (\ref{e2.3})-(\ref{e2.4})) of the Riccati difference equation (\ref{GRDE}). This result will play a key role in the proof of the proposed existence conditions. 

\begin{thm}\label{T5.1}
Let $\mathbb{X}^2(\cdot):\mathcal{I}_2\subset {\mathbb Z}_+\rightarrow \mathcal{S}_n^N$ be an admissible solution of a Riccati difference equation of type (\ref{GRDE}) associated to the quadruple $\sum^2=\left(\{{\bf A}(t)\}_{t\geq 0}, \{{\bf B}(t)\}_{t\geq 0}, \{{P_t}\}_{t\geq 0},\{{\mathbb Q}^2(t)\}_{t\geq 0} \right)$  and $\mathbb{X}^1(\cdot):\mathcal{I}_1\subset {\mathbb Z}_+\rightarrow \mathcal{S}_n^N$ be a solution of (\ref{GRDE}) associated to the quadruple
$\sum^1=\big(\{{\bf A}(t)\}_{t\geq 0}, \{{\bf B}(t)\}_{t\geq 0}, \{{P_t}(t)\}_{t\geq 0},$ $\{{\mathbb Q}^1(t)\}_{t\geq 0} \big)$ that verifies the sign condition  (\ref{e2.4}). $\mathbf{A}(t)$, $\mathbf{B}(t)$, $P_t$ are as in Remark \ref{R2.1} while:
\begin{equation}\label{eq.4.1}
Q^l(t,i)=\left(\begin{array}{cc}M^l(t,i) & L^l(t,i) \\ \star & R^l(t,i)\end{array}\right),\;l=1,2.
\end{equation}
Assume that $Q^2(t,i)\geq Q^1(t,i)$ for all $t\in \mathcal{I}_1\cap \mathcal{I}_2$, $1\leq i\leq N$.
If there exists $\tau \in \mathcal{I}_1\cap \mathcal{I}_2$ such that $X^2(\tau,i)\geq X^1(\tau,i)$, $1\leq i\leq N$, then $X^2(t,i)\geq X^1(t,i)$, for all $t\in \left[0, \; \tau\right]\cap \left(\mathcal{I}_1\cap \mathcal{I}_2\right)$, $1\leq i\leq N$.\\
\end{thm}

\begin{proof} Since $\mathbb{X}^2(\cdot)$ is an admissible solution it follows that it satisfies the sign conditions of type (\ref{e2.3})-(\ref{e2.4}). Therefore, and taking into account Remark \ref{R4.2a}, the following factorization takes place:
\begin{equation}
R^2(t,i)+\Pi_3(t)[\mathbb{X}^2(t)](i)=\left(\begin{array}{cc}V_{11}^2(t,i) & \mathbf{0} \\V_{21}^2(t,i)& V_{22}^2(t,i)\end{array}\right)^T\left(\begin{array}{cc}-{\mathbb I}_{m_1} & \mathbf{0} \\\mathbf{0} & {\mathbb I}_{m_2}\end{array}\right)\star
\end{equation}
$R^2(t,i)$ being the $(2,2)$ block of (\ref{eq.4.1}) and $V_{jk}^2(t,i)$ stand for $V_{jk}(t)[\mathbb{X}^2(t)](i)$, $1\leq i\leq N$.\\
\noindent Let $\Gamma(t,i)=\left(\begin{array}{c}\Gamma_1(t,i) \\\Gamma_2(t,i)\end{array}\right)$, $\Gamma_k(t,i)\in \mathbb{R}^{m_k\times n}$, $k=1,2$, $1\leq i\leq N$. Applying Lemma A.1 from Appendix A, we deduce that the Riccati difference equation solved by $\mathbb{X}^l(\cdot)$, $l=1, 2$, may be rewritten in the form:
\begin{equation}\label{eq.4.3.}
X^l(t,i)=\mathcal{L}_\Gamma^*(t)[\mathbb{X}^l(t+1)](i)+\mathcal{Q}_\Gamma^l(t,i)-\left(\Gamma(t,i)-F^l(t,i)\right)^T\left(R^l(t,i)+\Pi_3(t)[\mathbb{X}^l(t)](i)\right)\star
\end{equation}
where:
\begin{equation}\label{eq.4.4.}
\mathcal{Q}^l_\Gamma(t,i)=\left(\begin{array}{cc} I_n & \Gamma^T(t,i)\end{array}\right)Q^l(t,i)\star
\end{equation}
\begin{equation*}
F^l(t,i)=-\left(R^l(t,i)+\Pi_3(t)[\mathbb{X}^l(t+1)](i)\right)^{-1}\left(\Pi_2^T(t)[\mathbb{X}^l(t+1)](i)+(L^l(t,i))^T\right).
\end{equation*}
Let:
\begin{equation}\label{eq.4.6.bis}
\begin{cases}
F_1^l(t,i)=\left(\begin{array}{cc} I_{m_1} & \mathbf{0}\end{array}\right)F^l(t,i)\\
F_2^l(t,i)=\left(\begin{array}{cc}\mathbf{0} & I_{m_2}\end{array}\right)F^l(t,i)
\end{cases},\;l=1,2,\;1\leq i\leq N
\end{equation}
Taking:
\begin{equation}\label{eq.4.6.bis1}
\begin{cases}
\Gamma_1(t,i)=F^1_1(t,i)\\
\Gamma_2(t,i)=F^2_2(t,i)-\left(V_{22}^2(t,i)\right)^{-1}V_{21}^2(t,i)\left(F_1^1(t,i)-F_1^2(t,i)\right)
\end{cases}
,\;1\leq i\leq N
\end{equation}
and employing (\ref{eq.4.3.}) and (\ref{eq.4.6.bis})-(\ref{eq.4.6.bis1}) one obtains that: $t\rightarrow \mathbb{X}^2(t)-\mathbb{X}^1(t)$ solves the backward difference equation:
\begin{equation}\label{eq.4.7.}
\mathbb{X}^2(t)-\mathbb{X}^1(t)=\mathcal{L}_\Gamma^*(t)\left[\mathbb{X}^2(t+1)-\mathbb{X}^1(t+1)\right]+\mathbb{H}(t),\; t\in \mathcal{I}_1\cap \mathcal{I}_2
\end{equation}
where $\mathbb{H}(t)=\left(\begin{array}{ccc}H(t,1), & \cdots, & H(t,N)\end{array}\right)$ with:
\begin{align}\label{eq.4.8.}
H(t,i)&=\mathcal{Q}_\Gamma^2(t,i)-\mathcal{Q}_\Gamma^1(t,i)+\left(F_1^1(t,i)-F_1^2(t,i)\right)^T\left(V_{11}^2(t,i)\right)^2\star\notag\\
&+\left(\Gamma_2(t,i)-F_2^1(t,i)\right)^T\left(R_{22}^1(t,i)+\Pi_{322}(t)\left[\mathbb{X}^1(t+1)\right](i)\right)\star,\;1\leq i\leq N.
\end{align}
The assumption $Q^2(t,i)\geq Q^1(t,i)$ together with (\ref{eq.4.4.}) yields: $\mathcal{Q}_\Gamma^2(t,i)- \mathcal{Q}_\Gamma^1(t,i)\geq 0$, $t\in \mathcal{I}_1\cap \mathcal{I}_2$, $1\leq i\leq N$.
Hence (\ref{eq.4.8.}) allows us to conclude that $H(t,i)\geq 0$, $t\in \mathcal{I}_1\cap \mathcal{I}_2$, $1\leq i\leq N$. So, by induction, the proof is completed. \hfill$\blacksquare$
\end{proof}

\noindent For each $\tau>0$ we consider $\mathbb{X}_\tau(t)=\left(\begin{array}{ccc}X_\tau(t,1), & \cdots & X_\tau(t,N)\end{array}\right)$ the solution of (\ref{GRDE}) satisfying the condition $X_\tau(\tau+1,i)=0$, $1\leq i\leq N$ and $\mathcal{I}(\tau)\subset\left[0, \tau+1\right]$ the maximal interval on which $\mathbb{X}_\tau(\cdot)$ is defined. Note that under assumption $\textbf{H4)}$, $R(t,i)$ is boundedly invertible for every $(t,i)\in \mathbb{Z}_+\times \mathfrak{N}$, hence  $\mathcal{I}(\tau)$ is not empty.\\

\noindent We are now in position to prove the following result.\\

\begin{prop}\label{P5.5}
Assume:
\begin{itemize}
\item[a)] The assumptions {\bf H1)} a), {\bf H2)}, {\bf H3)} a) and {\bf H4)} are fulfilled.
\item[b)] ${\cal A}^{\Sigma}$ is not empty.
\end{itemize}
Under these assumptions, the following hold:
\begin{itemize}
 \item [i)] For any $\tau>0$ the solution $\mathbb{X}_\tau(t)$ is well defined for all $0\leq t\leq \tau+1$.
\item [ii)] For any $\tau>0$ we have $0\leq X_\tau(t,i)\leq \hat{X}(t,i)$ for all $0\leq t\leq \tau+1$, $i\in \mathfrak{N}$, $\mathbb{\hat{X}}(t)=(\hat{X}(t,1),  \cdots  ,\hat{X}(t,N))$, $t\in \mathbb{Z}_+$ being an arbitrary admissible globally defined positive semi-definite solution of (\ref{GRDE}).
\item [iii)] For any $\tau>0$, $\mathbb{X}_\tau(t)$ satisfies the sign conditions
 (\ref{e2.3})-(\ref{e2.4}) for all $0\leq t\leq \tau+1$ with the constants $\delta_k>0$ not depending upon $\tau$.
\item [iv)] There exists $\mu>0$ such that $0\leq X_{\tau_1}(t,i)\leq X_{\tau_2}(t,i)\leq \mu I_n$ for all $0\leq t\leq \tau_1\leq \tau_2$, $i\in \mathfrak{N}$.
\item [v)] If $\tilde{X}(t,i)=\underset{\tau\rightarrow \infty}{\lim}X_\tau(t,i)$ then $\mathbb{\tilde{X}}(t)=(\tilde{X}(t,1),  \cdots, \tilde{X}(t,N))$, $t\geq 0$, is the admissible minimal positive semi-definite solution of (\ref{GRDE}).
\item [vi)] If $t\rightarrow A_j(t,i)$, $t\rightarrow B_j(t,i)$, $0\leq j\leq r$, $t\rightarrow P(t)$, $t\rightarrow L(t,i)$, $t\rightarrow M(t,i)$ and $t\rightarrow R(t,i)$ are $\mathfrak p$-periodic sequences  then $\tilde{\mathbb{X}}(\cdot)$ is also a $\mathfrak p$-periodic sequence.
\item [vii)] If $\mathcal{A}_0^{\Sigma}$ is not empty, then the solution $\tilde{\mathbb{X}}(\cdot)$ satisfies besides (\ref{e2.4}) the following sign condition:
\begin{equation}\label{Prop4.2_equa1}
R_{11}(t,i)+\sum_{k=1}^rB_{k1}^T(t,i)\tilde{X}(t,i)B_{k1}(t,i)\leq -\mu I_{m_1}
\end{equation}
for all $(t,i)\in [0,\tau+1]\times \mathfrak{N}$, $\tau>0$, $\mu>0$ being a constant independent upon $t,i,\tau$.
\end{itemize}
\end{prop}

\begin{proof}
\begin{itemize}
\item [i)] Let us assume that ${\cal I}(\tau)=\{t_0, t_0+1, ..., \tau+1\}$ with $t_0\geq 1$. Based on (\ref{e80a}) and (\ref{e80b}) it follows that $R(t_0-1,i)+\Pi_3(t_0-1)[{\mathbb X}_{\tau}(t_0)](i)$ are invertible. Hence, ${\mathbb X}_{\tau}(t_0-1)=(X_{\tau}(t_0-1,1), ..., X_{\tau}(t_0-1,N))$ is well defined by (\ref{GRDE}) written for $t=t_0-1$. In this way, we have shown that $t_0-1\in{\cal I}(\tau)$.
thus one obtains inductively that ${\mathbb X}_{\tau}(t)$ is well defined for $0\leq t\leq \tau+1$.
\item [ii)] It follows immediately from Theorem \ref{T5.1}.
\item [iii)] It follows directly from(\ref{e80a}) and(\ref{e80b}) written for $t_0=t+1$, $t=0,1,...,\tau$.
\item [iv)] The proof of $0\leq X_{\tau_1}(t,i)\leq X_{\tau_2}(t,i)$ is a direct consequence of the Comparison Theorem (Theorem \ref{T5.1}).
The existence of $\mu>0$ follows from Lemma C.3. (i) and the boundedness property of $\tilde{\mathbb{X}}_{\mathbb{K}\mathbb{W}}(\cdot)$.
\item [v)] The existence of the limit $\tilde{X}(t,i)=\underset{\tau\rightarrow \infty}{\lim}X_\tau(t,i)$ follows from iv) above.
The minimality and admissibility properties are direct consequences of assertion ii) and iii).
\item [vi)] The periodicity property can be proved using the same steps as in in the proof of Theorem 1 from \cite{SV:Automatica}.
\item [vii)] The result follows from Lemma C.3. i) and (\ref{sign_closed'}).
\end{itemize}\hfill$\blacksquare$
\end{proof}

\noindent Let us now introduce the following auxiliary system:
\begin{equation}\label{auxiliary}
\begin{cases}
x(t+1)={\check{A}}_0(t,\theta_t)x(t)+\sum_{j=1}^rw_j(t){\check{A}}_j(t,\theta_t)x(t)\\
y(t)={\check{C}}(t,\theta_t)x(t)
\end{cases}
\end{equation}
where
\begin{equation}\label{e3.5}
{\check{A}}_j(t,i)=A_j(t,i)-B_{j2}(t,i)R_{22}^{-1}(t,i)L_2^T(t,i),\;\; 0\leq j\leq r
\end{equation}
and ${\check{C}}(t,i)$ are obtained from the factorization $M(t,i)-L_2(t,i)R_{22}^{-1}(t,i)L_2^T(t,i)={\check{C}}^T(t,i){\check{C}}(t,i)$,
for all $i\in \mathfrak{N}$, $t\geq 0$.\\

\noindent The main result of this section is given as follows:\\
\begin{thm}\label{T5.6}
Assume that:
\begin{itemize}
\item [a)] the assumptions \textbf{H1})-\textbf{H4}) are fulfilled;
\item [b)] the set $\mathcal{A}^{\Sigma}$ is not empty;
\item [c)] there exists an initial distribution $\tilde{\pi}_0$ such that for the Markov chain $(\{\tilde{\theta}_t\}_{t\geq 0},\{P_t\}_{t\geq 0},\mathfrak{N})$ (with initial distribution $\tilde{\pi}_0$), there exists $\delta>0$ such that:
\begin{equation*}\label{H4}
\tilde{\pi}_t(i)=\mathrm{P} \{\tilde{\theta}_t=i\}\geq \delta,\; t\geq 0,\; i \in \mathfrak{N}.
\end{equation*}
\item[d)] the auxiliary system (\ref{auxiliary}) is stochastically detectable:
\end{itemize}
then:
\begin{itemize}
\item [i)] $\tilde{\mathbb{X}}(\cdot)$ defined in v) from Proposition \ref{P5.5} coincides with the unique admissible stabilizing solution ${\mathbb{X}}_s(\cdot)$ of (\ref{GRDE}).
\item [ii)] If $\mathcal{A}^{\Sigma}_0$ is not empty, then the solution $\tilde{\mathbb{X}}(\cdot)$ satisfies (\ref{e2.4}) together with the sign condition (\ref{Prop4.2_equa1}).
\end{itemize}
\end{thm}

\noindent In order to prove the above assertions, we need first to prove some auxiliary results.

\begin{lem}\label{L5.7}
 Assume that the assumptions of Theorem \ref{T5.6} hold.
If $\mathbb{X}(\cdot)$ is a bounded on $\mathbb{Z}_+$ positive semi-definite solution of (\ref{GRDE}) then the system:
\begin{align}\label{L3.4_equa6}
x(t+1)&=[A_0(t,\theta_t)+B_{02}(t,\theta_t)V_{22}^{-1}(t,\theta_t)V_2(t,\theta_t)F(t,\theta_t)
+\notag\\&+\sum_{k=1}^rw_k(t)(A_k(t,\theta_t)+B_{k2}(t,\theta_t)V_{22}^{-1}(t,\theta_t)V_2(t,\theta_t)F(t,\theta_t))]x(t)
\end{align}
is ESMS, where $F(t,i)$ is defined as in (\ref{e2.8}) with $\mathbb{X}_s(\cdot)$ being replaced by  ${\mathbb{X}}(\cdot)$ and $V_2(t,\theta_t)=\left[\begin{array}{cc}V_{21}(t,\theta_t) & V_{22}(t,\theta_t)\end{array}\right]$, $V_{jk}(t,\theta_t)=V_{jk}(t)[{\mathbb X}(t+1)](\theta_t)$ introduced in Remark \ref{R4.2a}.
\end{lem}
\begin{proof} Using Lemma A.1. in Appendix A with:
\begin{equation}\label{L3.4_equa3}
\begin{cases}
\Gamma_1(t,i)=0\\
\Gamma_2(t,i)=F_2(t,i)+V_{22}^{-1}(t,i)V_{21}(t,i)F_1(t,i)
\end{cases}
;\;(t, i)\in \mathbb{Z}_+\times \mathfrak{N}
\end{equation}
one shows that the Riccati difference equation (\ref{GRDE}) can be rewritten as:
\begin{align}\label{L3.4_equa4}
X(t,i)&=\sum_{k=0}^r(A_k(t,i)+B_{k2}(t,i)\Gamma_2(t,i))^T\Xi(t)[\mathbb{X}(t+1)](i)\star+\mathcal{M}_\Gamma(t,i)\notag\\
&-F_1^T(t,i)\left[\begin{array}{cc}I & -V_{21}^T(t,i)V_{22}^{-1}(t,i)\end{array}\right]\left[R(t,i)+\Pi_3(t)[\mathbb{X}(t+1)](i)\right]\left[\begin{array}{c}I \\-V_{22}^{-1}(t,i)V_{21}(t,i)\end{array}\right]F_1(t,i)
\end{align}
for $(t,i)\in \mathbb{Z}_+\times \mathfrak{N}$.
After some algebraic manipulations, one shows that (\ref{L3.4_equa4}) can be rewritten as:
\begin{align}\label{L3.4_equa5}
X(t,i)&=\sum_{k=0}^r(A_k(t,i)+B_{2k}(t,i)\Gamma_2(t,i))^T\Xi(t)[\mathbb{X}(t+1)](i)\star+\notag\\
&+F_1^T(t,i)V_{11}^T(t,i)V_{11}(t,i)F_1(t,i)+\check{C}^T(t,i)\check{C}(t,i)+\notag\\
&+\left[L_2(t,i)+\Gamma_2^T(t,i)R_{22}(t,i)\right]R_{22}^{-1}(t,i)\star
\end{align}
$(t,i)\in \mathbb{Z}_+\times \mathfrak{N}$.
Let $\mathcal{V}(t,i)=x^T(t)\tilde{X}(t,i)x(t)$ where $x(t)=x(t,t_0,x_0)$ is the solution of (\ref{L3.4_equa6})
with $x(t_0)=x_0$. Using conditional expectation properties, one can show that:
\begin{align}
\label{L3.4_equa7}
&\sum_{t=t_0}^\tau {\mathbb E}\left[\mathcal{V}(t+1,\theta_{t+1})-\mathcal{V}(t,\theta_t)\Big|\theta_{t_0}=i\right]
={\mathbb E}\left[\mathcal{V}(\tau+1,\theta_{\tau+1})\Big|\theta_{t_0}=i\right]-x_0^T\tilde{X}(t_0,i)x_0\notag\\
&=-\sum_{t=t_0}^\tau {\mathbb E} \left[\Big|V_{11}(t,\theta_t)F_{1}(t,\theta_t)x(t)\Big|^2+\Big|\check{C}(t,\theta_t)x(t)\Big|^2+\Big|R_{22}^{-\frac{1}{2}}(t,\theta_t)\bar{L}(t,\theta_t)x(t)\Big|^2\Big|\theta_{t_0}=i\right]
\end{align}
for all $\tau \geq t_0$, where $\bar{L}(t,\theta_t)=L_2(t,\theta_t)+\Gamma_2^T(t,\theta_t)R_{22}(t,\theta_t)$.\\
\noindent Since $\tilde{\mathbb{X}}(\cdot)$ is a semi-positive bounded solution we deduce that:
\begin{align}
\label{L3.4_equa8}
\sum_{t=t_0}^\infty {\mathbb E} \left[\Big|\check{C}(t,\theta_t)x(t)\Big|^2+\Big|R_{22}^{-\frac{1}{2}}(t,\theta_t)\bar{L}(t,\theta_t)x(t)\Big|^2\Big|\theta_{t_0}=i\right]\leq c|x_0|^2
\end{align}
for some positive constant $c>0$. \\
It follows from the detectability assumption that there exist gain matrices $\mathbb{H}(\cdot):\mathbb{Z}_+\rightarrow \mathcal{M}_{n,p}^N$ such that:
\begin{align}
\label{L3.4_equa9}
x(t+1)&=\left(\check{A}_0(t,\theta_t)+H(t,\theta_t)\check{C}(t,\theta_t)\right)x(t)+\sum_{k=1}^rw_k(t)\check{A}_k(t,\theta_t)x(t)
\end{align}
is ESMS. Note also that (\ref{L3.4_equa6}) can be rewritten as:
\begin{align}
\label{L3.4_equa10}
x(t+1)=\left(\check{A}_0(t,\theta_t)+H(t,\theta_t)\check{C}(t,\theta_t)\right)x(t)+f_0(t,\theta_t)+\sum_{k=1}^rw_k(t)\left(\check{A}_k(t,\theta_t)x(t)+f_k(t,\theta_t)\right)
\end{align}
where:
\begin{equation}
\label{L3.4_equa11}
\begin{cases}
f_0(t,i)=-H(t,i)\check{C}(t,i)x(t)+B_{02}(t,i)R_{22}^{-1}(t,i)\bar{L}^T(t,i)x(t)\\
f_k(t,i)=B_{k2}(t,i)R_{22}^{-1}(t,i)\bar{L}^T(t,i)x(t)
\end{cases}
\end{equation}
$(t,i)\in \mathbb{Z}_+\times \mathfrak{N}$.
Hence reasoning as in the proof of Corollary 3.9 in \cite{carte2010} one deduces that there exist a positive constant $\nu$ not depending upon $(t,x_0,i)\in \mathbb{Z}_+\times \mathbb{R}^n\times \mathfrak{N}$ such that:
\begin{equation}
\sum_{t=t_0}^\infty {\mathbb E}\left[|x(t)|^2|\theta_{t_0}=i\right]\leq \nu |x_0|^2
\end{equation}
for all $(t,x_0,i)\in \mathbb{Z}_+\times \mathbb{R}^n\times \mathfrak{N}$.
Using the implication $ii)\rightarrow i)$ from Theorem 3.2 in \cite{Dragan:14}, one obtains that system (\ref{L3.4_equa6}) is ESMS-CI. Using Theorem 3.6 from \cite{carte2010} one shows that under assumption c) from the statement of Theorem \ref{T5.6} from above, that the stability notions ESMS-CI and ESMS are equivalent. This ends the proof. \hfill$\blacksquare$
\end{proof}
\begin{rem}\label{Rem5.1}
In Lemma \ref{L5.7}, if one stresses the assumption c) from the statement of Theorem \ref{T5.6} then we can still prove that system (\ref{L3.4_equa6}) is ESMS. This is because, under the considered assumptions, ESMS-CI implies ESMS (please see Theorem 3.3 from \cite{carte2010}) .
\end{rem}

\begin{lem}\label{L5.8}
 Assume that the assumptions a)- c) from Theorem \ref{T5.6} are fulfilled. Let  $\tilde{X}(t,i)=\underset{\tau\rightarrow \infty}{\lim}X_\tau(t,i)$, $(t,\;i) \in \mathbb{Z}_+\times \mathfrak{N}$ and $\tilde{F}(t,i)$ is defined as in (\ref{e2.8}) with $\mathbb{X}_s(\cdot)$ being replaced by  $\tilde{\mathbb{X}}(\cdot)$.
Then there exists $c>0$ not depending upon $t_0$, such that:
\begin{equation}\label{eq_77}
{\mathbb E}\left[\sum_{t=t_0}^\infty\tilde{\Phi}^T(t,t_0)\tilde{F}_1^T(t,\theta_t)\tilde{F}_1(t,\theta_t)\tilde{\Phi}(t,t_0)\Big|\theta(t_0)=i\right]\leq c I_n,\; i\in \mathfrak{N}
\end{equation}
where $\tilde{\Phi}(t,t_0)$ is the fundamental random matrix solution of the system:
\begin{equation}\label{eq_78}
x(t+1)=\left[A_0(t,\theta_t)+B_0(t,\theta_t)\tilde{F}(t,\theta_t)+\sum_{k=1}^rw_k(t)\left(A_k(t,\theta_t)+
B_k(t,\theta_t)\tilde{F}(t,\theta_t)\right)\right]x(t).
\end{equation}
\end{lem}

\begin{proof}
\noindent We shall show that for each $i\in\mathfrak{N}$ we have:
\begin{equation}\label{supre}
\sup\left\{\Big|{\mathbb E} \left[\sum_{t=t_0}^T\tilde{\Phi}^T(t,t_0)\tilde{F}_1^T(t,\theta_t)\tilde{F}_1(t,\theta_t)\tilde{\Phi}(t,t_0)\Big|\theta(t_0)=i\right]\Big|:t_0\geq 0,\;T>t_0\right\}<\infty.
\end{equation}
Assume by contrary that there exists $i_0\in\mathfrak{N}$ such that the supremum in (\ref{supre}) is not finite. Hence, following step by step the reasoning in the proof of Lemma 3 from \cite{Dragan:18} one shows that there exist sequences $\tau_k$, $t_k$, $\tau_k>t_k$ and $k_0\geq 1$ such that:
\begin{equation}
{\mathbb E}\left[\sum_{t=t_k}^{\tau_k}|{F}_1^{\tau_k}(t,\theta_t)\tilde{x}_{\tau_k}(t)|^2\Big|\theta_{t_k}=i\right]\geq k,\; i\in \mathfrak{N}
\end{equation}
for every $k\geq k_0$, ${F}_1^{\tau_k}(t,\theta_t)$ being defined as in (\ref{equa_F}) and $\tilde{x}_{\tau_k}(t)={\Phi}^{\tau_k}(t,t_k)x_k$, $t\in [t_k,  \tau_k]$ where ${\Phi}^{\tau_k}(t,t_k)$ is the fundamental random matrix solution of the system:
\begin{equation}
x(t+1)=\left[A_0(t,\theta_t)+B_0(t,\theta_t){F}_{\tau_k}(t,\theta_t)+\sum_{k=1}^rw_k(t)\left(A_k(t,\theta_t)+B_k(t,\theta_t){F}_{\tau_k}(t,\theta_t)\right)\right]x(t)
\end{equation}
and $x_k\in \mathbb{R}^n$, $|x_k|=1$. This leads to:
\begin{equation}\label{Lemma4.5_equa1}
{\mathbb E}\left[\sum_{t=t_k}^{\tau_k}|{F}_1^{\tau_k}(t,\theta_t)\tilde{x}_{\tau_k}(t)|^2\right]\geq \pi_{t_k}(i)k,\; i\in \mathfrak{N}.
\end{equation}
The condition (\ref{Lemma4.5_equa1}) written for the initial distribution $\tilde{\pi}_0$ yields (using assumption c) from the statement of Theorem \ref{T5.6}):
\begin{equation}\label{Lemma4.5_equa2}
{\mathbb E}\left[\sum_{t=t_k}^{\tau_k}|{F}_1^{\tau_k}(t,\tilde{\theta}_t)\tilde{x}_{\tau_k}(t)|^2\right]\geq \delta k,\; i\in \mathfrak{N}.
\end{equation}
Let $\tilde{u}_1^{\tau_k}(t)={F}_1^{\tau_k}(t,\tilde{\theta}_t)\tilde{x}_{\tau_k}(t)$, $t\in [t_k, \tau_k]$, hence it follows from Lemma C.1. ii) that:
\begin{equation}
\sum_{i=1}^N\pi_{t_k}(i)x_k^TX_{\tau_k}(t_k,i)x_k\leq \sum_{i=1}^N\pi_{t_k}(i)\mathcal{J}(t_k,\tau_k,x_k,i,\tilde{u}_1^{\tau_k},u_2)
\end{equation}
for all $k\geq k_0$, $i\in \mathfrak{N}$ and $u_2\in\ell_{\tilde{\mathcal{H}}}^2\{t_k,\tau_k;\mathbb{R}^{m_2}\}$.
Let $(\mathbb{K}(\cdot), \mathbb{W}(\cdot))\in \mathcal{A}^{\Sigma}$ be arbitrary but fixed and let $u_2^k$ be defined by $u_2^k(t)=K(t,\theta_t)\chi_k(t)+W(t,\theta_t)\tilde{u}_1^{\tau_k}(t)$ where $\chi_k(t)$, $t_k\leq t\leq \tau_k$ is the solution of the following IVP:
\begin{equation}
\begin{cases}
\chi_k(t+1)=\left[A_{0{\mathbb K}}(t,\theta_t)+\sum_{k=1}^rw_k(t)A_{k{\mathbb K}}(t,\theta_t)\right]\chi_k(t)+\\
\quad\quad\quad\quad\quad+\left[B_{0{\mathbb W}}(t,\theta_t)+\sum_{k=1}^rw_k(t)B_{k{\mathbb W}}(t,\theta_t)\right]\tilde{u}_1^{\tau_k}(t)\\
\chi_k(t_k)=x_k.
\end{cases}
\end{equation}
It follows from Proposition B.4. that:
\begin{equation}
\sum_{i=1}^N\pi_{t_k}(i)\mathcal{J}(t_k,\tau_k,x_k,i,\tilde{u}_1^{\tau_k},u_2^k)\leq \alpha_0  -\alpha_1{\mathbb E} \left[\sum_{t=t_k}^{\tau_k}|\tilde{u}_1^{\tau_k}(t)|^2\right]
\end{equation}
for some positive constants $\alpha_0$ and $\alpha_1$ which do not depend upon $k$. We deduce from the above developments that:
\begin{equation}
\label{equa_82}
\sum_{i=1}^N\pi_{t_k}(i)x_k^TX_{\tau_k}(t_k,i)x_k<0
\end{equation}
if $\delta k>\max\{\delta k_0,\;\frac{\alpha_0}{\alpha_1}\}$.
But (\ref{equa_82}) is in contradiction with the result given in Lemma C.3. Thus we have obtained that (\ref{eq_77}) is true for the Markov chain $(\{\tilde \theta_t\}_{t\geq 0}, \{P_t\}_{t\geq 0}, {\mathfrak N})$, that is: 
\begin{eqnarray}\label{1000}
\sum_{t=t_0}^{\infty} {\mathbb E}[\tilde{\tilde{\Phi}}^T(t,t_0)\tilde F_1^T(t,\tilde\theta_t)\tilde F_1(t,\tilde\theta_t)\tilde{\tilde{\Phi}}(t,t_0)| \tilde\theta_{t_0}=i]\leq c I_n
\end{eqnarray}
for all $t_0\in{\mathbb Z}_+, i\in{\mathfrak N}$, where $\tilde{\tilde{\Phi}}(t,t_0)$ is the fundamental random matrix solution of the system obtained from (\ref{eq_78}) written for $\tilde\theta_t$ instead of $\theta_t$.

\noindent On the other hand from the representation theorem (Theorem 3.1 in \cite{carte2010}) and Remark 3.2 from \cite{carte2010} we deduce that
\begin{eqnarray}\label{1001}
{\mathbb E}[\tilde{\tilde{\Phi}}^T(t,t_0)\tilde F_1^T(t,\tilde\theta_t)\tilde F_1(t,\tilde\theta_t)\tilde{\tilde{\Phi}}(t,t_0)|\tilde\theta_{t_0}=i]=\nonumber\\
{\mathbb T}_{\tilde F}^*(t,t_0)[{\mathbb H}(t)](i)=
{\mathbb E}[\tilde\Phi^T(t,t_0)\tilde F_1^T(t,\theta_t)\tilde F_1(t,\theta_t)\tilde\Phi(t,t_0)|\theta_{t_0}=i]
\end{eqnarray}
for all $t\geq t_0\geq 0, i\in{\mathfrak N}$ and for arbitrary Markov chain $(\{\theta_t\}_{t\geq 0}, \{P_t\}_{t\geq 0}, {\mathfrak N})$, where ${\mathbb T}_{\tilde F}^*(t,t_0)$ is the adjoint operator of the linear evolution operator on ${\cal S}_n^N$ defined by the discrete time linear equation  (\ref{e2.7}) written for $\tilde F(t,i)$ instead of $F_s(t,i)$ and ${\mathbb H}(t)=(H(t,1),\cdots, H(t,N))$ with $H(t,i)=\tilde F_1^T(t,i)\tilde F_1(t,i)$.

\noindent Now (\ref{1000}) and (\ref{1001}) alow us to conclude that (\ref{eq_77}) holds for any Markov chain $(\{\theta_t\}_{t\geq 0}, \{P_t\}_{t\geq 0}, {\mathfrak N})$. Thus the proof is complete.\hfill$\blacksquare$ 
\end{proof}

\begin{rem} In \cite{carte2010} the representation formula (3.4) is proved for $H(\theta_t)$ and not for $H(t,\theta_t)$.  One can show that an equality of type (3.4) from \cite{carte2010} takes place also when $H(i)$ is replaced by $H(t,i), t\geq t_0, i\in{\mathfrak N}$. Indeed we chose $t_1\geq t_0$ arbitrary but fixed. Applying (3.4) from \cite{carte2010} we obtain 
\begin{eqnarray}\label{1002}
{\mathbb E}[\tilde \Phi^T(t,t_0)H(t_1,\theta_t)\tilde\Phi(t,t_0)| \theta_{t_0}=i ]={\mathbb T}_{\tilde F}^*(t,t_0)[{\mathbb H}(t_1)](i)
\end{eqnarray}
for all $t\geq t_0, i\in{\mathfrak N}$. Particularly, (\ref{1002}) is  also true for $t=t_1$ which leads to the desired equality.
\end{rem}

\noindent \textbf{Proof of Theorem \ref{T5.6}}
\begin{itemize}
\item [i)] Let  $\tilde{X}(t,i)=\underset{\tau\rightarrow \infty}{\lim}X_\tau(t,i)$, $(t,\;i) \in \mathbb{Z}_+\times \mathfrak{N}$.
It follows from Proposition \ref{P5.5} that $\tilde{\mathbb{X}}(\cdot)$ is a bounded on $\mathbb{Z}_+$ positive semi-definite solution of equation (\ref{GRDE}). Let us now show that $\tilde{\mathbb{X}}(\cdot)$ is a stabilizing solution of (\ref{GRDE}).
 First note that the system:
\begin{align}
x(t+1)=\left[A_0(t,\theta_t)+B_{0}(t,\theta_t)\tilde{F}(t,\theta_t)\right]x(t)+\sum_{k=1}^rw_k(t)\left[A_k(t,\theta_t)+B_{k}(t,\theta_t)\tilde{F}(t,\theta_t)\right]x(t)
\end{align}
(where $\tilde{F}(t,i)$ is defined as in (\ref{e2.8}) with $\mathbb{X}_s(\cdot)$ being replaced by  $\tilde{\mathbb{X}}(\cdot)$) can be rewritten as:
\begin{align}
x(t+1)=\hat{A}_0(t,\theta_t)x(t)+g_0(t,\theta_t)+\sum_{k=1}^rw_k(t)\left(\hat{A}_k(t,\theta_t)x(t)+g_k(t,\theta_t)\right)
\end{align}
where:
\begin{equation}
\begin{cases}
\hat{A}_j(t,\theta_t)=A_j(t,\theta_t)+B_{j2}(t,\theta_t)V_{22}^{-1}(t,\theta_t)V_2(t,\theta_t)\tilde{F}(t,\theta_t)\\
g_j(t,\theta_t)=\left[B_{j1}(t,\theta_t)-B_{j2}(t,\theta_t)V_{22}^{-1}(t,\theta_t)V_{21}(t,\theta_t)\right]\tilde{F}_1(t,\theta_t)x(t)
\end{cases}
;\; 0\leq j\leq r.
\end{equation}
It follows from Lemma \ref{L5.8} that $E\left[\sum_{t=t_0}^\infty |g_j(t,\theta_t)|^2\Big|\theta_{t_0}=i\right]\leq c_j|x_0|^2$.
Hence using the same reasoning as in the proof of Lemma \ref{L5.7} one shows that $\tilde{\mathbb{X}}(\cdot)$ is a stabilizing solution of (\ref{GRDE}). The unicity of $\tilde{\mathbb{X}}(\cdot)$ follows from Theorem 3.4.
\item[ii)] The proof is the same as for the proof of point vii) in Proposition \ref{P5.5}.
\end{itemize}
This ends the proof.\hfill$\blacksquare$

\begin{rem}
\begin{itemize}
\item [i)] The condition that ${\cal A}^{\Sigma}$ is not empty is also necessary. Indeed, one can prove that if the Riccati equation (GRDE) has a positive semidefinite and stabilizing solution, then the set ${\cal A}^{\Sigma}$ is not empty. To this end, one can takes $\tilde K(t,i)$ and $\tilde W(t,i)$ as defined in (\ref{e5.15})-(\ref{factor}) and prove that $(\tilde K(t,i),\;\tilde W(t,i))\in {\cal A}^{\Sigma}$.
\item[ii)] The stochastic detectability of the auxiliary system (\ref{auxiliary}) may be tested by the solvability of a set of LMIs conditions (The reader can refer to \cite{carte2010} for example). This shows, in addition with remark 4.4 i), the numerical tractability of the proposed set of existence conditions thanks to the LMI formalism.
\end{itemize}
\end{rem}

\section{A zero-sum stochastic LQ difference game}

In this section, we show the relation between the bounded and stabilizing solution of GDTRE (\ref{GRDE}) and the design of the feedback gains of an equilibrium strategy in the case of a zero-sum stochastic LQ difference game on infinite time horizon described by a controlled system of type (\ref{equa1}) and a quadratic performance criterion of type (\ref{cost}).\\
\noindent Let us first rewrite (\ref{equa1}) and (\ref{cost}) according to the partition $u=\left(\begin{array}{cc}u_1^T & u_2^T\end{array}\right)^T$ of the input $u$. Thus we obtain:
\begin{equation}\label{e5.1}
\begin{cases}
x(t+1)=A_0(t,\theta_t)x(t)+B_{01}(t,\theta_t)u_1(t)+B_{02}(t,\theta_t)u_2(t)\\
\quad\quad\quad +\sum_{k=1}^r\left[A_k(t,\theta_t)x(t)+B_{k1}(t,\theta_t)u_1(t)+B_{k2}(t,\theta_t)u_2(t)\right]w_k(t)\\
x(t_0)=x_0
\end{cases}
\end{equation}
\begin{equation}\label{e5.2}
\mathcal{J}(x_0,u_1(\cdot),u_2(\cdot))={\mathbb E}\left[\sum_{t_0}^\infty \left(\begin{array}{c}x_u(t) \\u_1(t) \\u_2(t)\end{array}\right)\left(\begin{array}{ccc}M(t,\theta_t) & L_1(t,\theta_t) & L_2(t,\theta_t) \\\star & R_{11}(t,\theta_t) & R_{12}(t,\theta_t) \\\star & \star & R_{22}(t,\theta_t)\end{array}\right)\star \right]
\end{equation}
In game theory terminology, the inputs $u_k:\mathbb{Z}_+\rightarrow \mathbb{R}^{m_k}$ are usually called strategies available for player $k$, $k=1,\;2$. These are functions with the property that $\left(\begin{array}{cc}u_1^T(\cdot) & u_2^T(\cdot)\end{array}\right)^T$ are in a subset of $\mathcal{U}_{\text{ad}}$. 
Such a subset of $\mathcal{U}_{\text{ad}}$ will be called the set of admissible strategies. 
The shape of this set depends on the type of available informations for player $k$ in order to compute the value of $u_k(t)$. 
In this section, we shall consider two different types of admissible strategies. 
Before detailing these two different strategies, we give here for the reader convenience a definition of an equilibrium strategy of the zero-sum difference game described above.

\begin{definition}\label{D6.1}
An admissible strategy $\left(\tilde{u}_1(\cdot), \tilde{u}_2(\cdot)\right)$ is named equilibrium strategy of the zero-sum differential game described by the controlled system (\ref{e5.1}) and the performance criterion (\ref{e5.2}) if:
\begin{equation}\label{e5.3}
\mathcal{J}(x_0,u_1(\cdot),\tilde{u}_2(\cdot))\leq \mathcal{J}(x_0,\tilde{u}_1(\cdot),\tilde{u}_2(\cdot))\leq \mathcal{J}(x_0,\tilde{u}_1(\cdot),u_2(\cdot))
\end{equation}
for all $u_1(\cdot)$, $u_2(\cdot)$ with the property that $(\tilde{u}_1(\cdot), {u}_2(\cdot))$ and $({u}_1(\cdot), \tilde{u}_2(\cdot))$ are admissible strategies.
\end{definition}

Roughly speaking, from (\ref{e5.3}) one sees that the aim of player 1 is to maximize the performance $\mathcal{J}(x_0,\cdot,\cdot)$ while the goal of player 2 is to minimize the value of $\mathcal{J}(x_0,\cdot,\cdot)$.

\subsection{The full state-feedback case}

In this case, the admissible strategies are constructed in a linear state-feedback form as follows:
\begin{equation}\label{e5.4}
u_k(t)=F_k(t,\theta_t)x(t),\; k=1,2.
\end{equation}
We denote $\mathcal{F}_{ad}$ the set of the pairs $(\mathbb{F}_1(\cdot),\;\mathbb{F}_2(\cdot))$ where $\mathbb{F}_k(\cdot)=(F_k(\cdot,1) , \cdots , F_k(\cdot,N))$ with $t\rightarrow F_k(t,i):\mathbb{Z}_+\rightarrow \mathbb{R}^{m_k\times n}$ are bounded and matrix valued sequences with the property that the closed-loop system:
\begin{equation}\label{e5.4'}
\begin{cases}
x(t+1)=\left[A_0(t,\theta_t)+B_{01}(t,\theta_t)F_1(t,\theta_t)+B_{02}(t,\theta_t)F_2(t,\theta_t)\right]x(t)\\
+\sum_{k=1}^rw_k(t)\left[A_k(t,\theta_t)+B_{k1}(t,\theta_t)F_1(t,\theta_t)+B_{k2}(t,\theta_t)F_2(t,\theta_t)\right]x(t)
\end{cases}
\end{equation}
is ESMS. In order to implement a strategy of the form (\ref{e5.4}), we assume that each player has access to the current Markov state $\theta_t$ and the current system state $x(t)$. One can easily verify that the strategies of type (\ref{e5.4}) lie in $\mathcal{U}_{ad}$.\\

\begin{rem}\label{R6.1} Obviously, there exists a one to one correspondence between the strategies (\ref{e5.4}) and the set of feedback gains $\mathcal{F}_{ad}$.
Hence, in the rest of this subsection we shall denote $\mathcal{J}(x_0,\mathbb{F}_1(\cdot),\mathbb{F}_2(\cdot))$ instead of $\mathcal{J}(x_0,u_1(\cdot),u_2(\cdot))$.
In this case, we shall say that $(\tilde{u}_1(\cdot),\;\tilde{u}_2(\cdot))$ is an equilibrium strategy (or equivalently, the pair of feedback gains $(\tilde{\mathbb{F}}_1(\cdot),\; \tilde{\mathbb{F}}_2(\cdot))$ provide an equilibrium strategy) if:
\begin{equation}\label{e5.3'}
\mathcal{J}(x_0,\mathbb{F}_1(\cdot),\tilde{\mathbb{F}}_2(\cdot))\leq \mathcal{J}(x_0,\tilde{\mathbb{F}}_1(\cdot),\tilde{\mathbb{F}}_2(\cdot))\leq \mathcal{J}(x_0,\tilde{\mathbb{F}}_1(\cdot),\mathbb{F}_2(\cdot))
\end{equation}
for arbitrary $(\mathbb{F}_1(\cdot),\; \mathbb{F}_2(\cdot))$ with the property that $(\mathbb{F}_1(\cdot),\; \tilde{\mathbb{F}}_2(\cdot))$ and $(\tilde{\mathbb{F}}_1(\cdot),\; \mathbb{F}_2(\cdot))$ lie in $\mathcal{F}_{ad}$.\\
\end{rem}

\begin{thm}\label{T6.1}
Assume that:
\begin{itemize}
\item [a)] Assumptions \textbf{H1}), \textbf{H2}) and \textbf{H3}) a) are fulfilled;
\item [b)] The GDTRE (\ref{GRDE}) associated to the system (\ref{e5.1}) and the quadratic functional (\ref{e5.2}) admits a bounded and stabilizing solution $\tilde{\mathbb{X}}(\cdot)$ satisfying the sign conditions:
\end{itemize}
\begin{equation}\label{e5.6}
\mathbb{R}_{11}(t,\tilde{\mathbb{X}}(t+1),i)\triangleq R_{11}(t,i)+\sum_{k=0}^rB_{k1}^T(t,i)\Xi(t)[\tilde{\mathbb X}(t+1)](i)B_{k1}(t,i)\leq-\mu_1I_{m_1}<0
\end{equation}
and
\begin{equation}\label{e5.7}
\mathbb{R}_{22}(t,\tilde{\mathbb{X}}(t+1),i)\triangleq R_{22}(t,i)+\sum_{k=0}^rB_{k2}^T(t,i)\Xi(t)[\tilde{\mathbb X}(t+1)](i)B_{k2}(t,i)\geq\mu_2I_{m_2}>0
\end{equation}
for all $(t,i)\in \mathbb{Z}_+\times \mathfrak N$.
Let $\tilde{F}_k(t,i)$ be defined as:
\begin{equation}\label{e5.5'}
\begin{cases}
\tilde{F}_1(t,i)=\left(\begin{array}{cc}I_{m_1} & 0\end{array}\right)\tilde{F}(t,i),\\
\tilde{F}_2(t,i)=\left(\begin{array}{cc}0 & I_{m_2}\end{array}\right)\tilde{F}(t,i)
\end{cases}
\end{equation}
$\tilde{F}(t,i)$ being computed via (\ref{e2.8}) (where $\mathbb{X}_s(\cdot)$ is replaced by $\tilde{\mathbb{X}}(\cdot)$).
Consider the strategies:
\begin{equation}\label{e5.8}
\tilde{u}_k(t)\triangleq \tilde{F}_k(t,\theta_t)x(t),\; t\geq 0,\; k=1,2.
\end{equation}
Under the considered assumptions $(\tilde{u}_1(\cdot),\;\tilde{u}_2(\cdot))$ is an equilibrium strategy for the zero-sum LQ differential game described by the controlled system (\ref{e5.1}), the quadratic functional (\ref{e5.2}) and the class of admissible strategies in the form of state-feedback as in (\ref{e5.4}).
\end{thm}

\begin{proof}
First, by using some standard arguments one can shows that under the considered assumptions one obtains:
\begin{equation}\label{equa_App_B_1}
\mathcal{J}(x_0,u_1(\cdot),u_2(\cdot))=x_0^T {\mathbb E} \left[\tilde{X}(t_0,\theta_{t_0})\right]x_0+{\mathbb E}\left[\sum_{t=t_0}^\infty\left(\begin{array}{c}u_1(t)-\tilde{F}_1(t,\theta_t)x(t) \\u_2(t)-\tilde{F}_2(t,\theta_t)x(t)\end{array}\right)^T\mathbb{R}(t,\tilde{\mathbb{X}}(t+1),\theta_t)\star \right]
\end{equation}
for all $u(\cdot)=(u_1(\cdot),\;u_2(\cdot))\in \mathcal{U}_{ad}$ and $\forall x_0\in \mathbb{R}^n$.
Note also that $(\tilde{\mathbb{F}}_1(\cdot),\;\tilde{\mathbb{F}}_2(\cdot))$ lies in $\mathcal{F}_{ad}$. Indeed,  $(\tilde{\mathbb{F}}_1(\cdot),\;\tilde{\mathbb{F}}_2(\cdot))$ being associated to the stabilizing solution $\tilde{\mathbb{X}}(\cdot)$, the result follows from Remark 2.3 and Theorem 3.3 in \cite{carte2010}. From (\ref{equa_App_B_1}) one obtains that:
\begin{equation}\label{e5.9}
\mathcal{J}(x_0,\tilde{\mathbb{F}}_1(\cdot),\tilde{\mathbb{F}}_2(\cdot))=x_0^T{\mathbb E}[\tilde{X}(t_0,\theta_{t_0})]x_0.
\end{equation}
Let $\mathbb{F}_k(\cdot)$ be arbitrary with the property that $({\mathbb{F}}_1(\cdot),\;\tilde{\mathbb{F}}_2(\cdot))$ and $(\tilde{\mathbb{F}}_1(\cdot),\;{\mathbb{F}}_2(\cdot))$ belong to $\mathcal{F}_{ad}$.\\
 We take:
\begin{equation}
\begin{cases}
u_1(t)=F_1(t,\theta_t)\hat{x}(t)\\
u_2(t)=\tilde{F}_2(t,\theta_t)\hat{x}(t)
\end{cases}
\end{equation}
where $\hat{x}(t),\; t\geq 0$ is the solution of the system of type (\ref{e5.4'}) written for $F_2(t,\theta_t)$ replaced by $\tilde{F}_2(t,\theta_t)$. We obtain via (\ref{equa_App_B_1}) that:
\begin{equation}\label{e5.10}
\mathcal{J}(x_0,{\mathbb{F}}_1(\cdot),\tilde{\mathbb{F}}_2(\cdot))=x_0^T{\mathbb E}[\tilde{X}(t_0,\theta_{t_0})]x_0+{\mathbb E}\left[\sum_{t=t_0}^\infty \hat{x}(t)\left(F_1(t,\theta_t)-\tilde{F}_1(t,\theta_t)\right)^T\mathbb{R}_{11}(t,\tilde{\mathbb{X}}(t+1),\theta_t)\star\hat{x}(t)\right].
\end{equation}
From (\ref{e5.6}), (\ref{e5.9}) and (\ref{e5.10}) we deduce that $\mathcal{J}(x_0,{\mathbb{F}}_1(\cdot),\tilde{\mathbb{F}}_2(\cdot))\leq \mathcal{J}(x_0,\tilde{\mathbb{F}}_1(\cdot),\tilde{\mathbb{F}}_2(\cdot))$, for all $\mathbb{F}_1(\cdot)$ such that $({\mathbb{F}}_1(\cdot),\;\tilde{\mathbb{F}}_2(\cdot))\in \mathcal{F}_{ad}$. Further, we take:
\begin{equation}
\begin{cases}
u_1(t)=\tilde{F}_1(t,\theta_t)\check{x}(t)\\
u_2(t)={F}_2(t,\theta_t)\check{x}(t)
\end{cases}
\end{equation}
where $\check{x}(t)$, $t \geq 0$ is the solution of the system of type (\ref{e5.4'}) written for $F_1(t,\theta_t)$ replaced by $\tilde{F}_1(t,\theta_t)$. We obtain via (\ref{equa_App_B_1}) that:
\begin{equation}\label{e5.11}
\mathcal{J}(x_0,\tilde{\mathbb{F}}_1(\cdot),{\mathbb{F}}_2(\cdot))=x_0^T{\mathbb E}[\tilde{X}(t_0,\theta_{t_0})]x_0+{\mathbb E}\left[\sum_{t=t_0}^\infty \check{x}(t)\left(F_2(t,\theta_t)-\tilde{F}_2(t,\theta_t)\right)^T\mathbb{R}_{22}(t,\tilde{\mathbb{X}}(t+1),\theta_t)\star\check{x}(t)\right].
\end{equation}
From (\ref{e5.7}), (\ref{e5.9}) and (\ref{e5.11}) we obtain that $\mathcal{J}(x_0,\tilde{\mathbb{F}}_1(\cdot),\tilde{\mathbb{F}}_2(\cdot))\leq \mathcal{J}(x_0,\tilde{\mathbb{F}}_1(\cdot),{\mathbb{F}}_2(\cdot))$, for all $x_0\in \mathbb{R}^n$ and $\mathbb{F}_2(\cdot)$ such that $(\tilde{\mathbb{F}}_1(\cdot),\;{\mathbb{F}}_2(\cdot))\in \mathcal{F}_{ad}$. Hence, $(\tilde{\mathbb{F}}_1(\cdot),\;\tilde{\mathbb{F}}_2(\cdot))$ provides an equilibrium strategy via (\ref{e5.8}). Thus the proof is complete.\hfill$\blacksquare$
\end{proof}

\begin{rem}
Note that conditions for the existence of a bounded and stabilizing solution $\tilde{\mathbb{X}}(\cdot)$ of GDTRE (\ref{GRDE}) satisfying the sign conditions (\ref{e5.6}) and (\ref{e5.7}) are given by Theorem \ref{T5.6}.
\end{rem}

\subsection{The full information feedback case}
We consider here the case when the player 1 has access to the strategies $u_1(\cdot)\in\ell_{\tilde{\mathcal{H}}}^2\{0,\infty;\mathbb{R}^{m_1}\}$, while the player 2 has access to the strategies of the form:
\begin{equation}\label{e5.12}
u_2(t,u_1(\cdot))=\varphi(t,x(t),u_1(t),\theta_t),\; t\geq t_0
\end{equation}
where the functions $(t,\;x,\;u_1)\rightarrow \varphi(t,x,u_1,i)$ have the properties:
\begin{itemize}
\item [$\alpha$)] $\varphi(t,0,0,i)=0,\; \forall t\in \mathbb{Z}_+,\; i\in \mathfrak N$;
\item [$\beta$)] There exists $\lambda>0$ not depending upon $(t,x,u_1)$ such that $|\varphi(t,x',u_1',i)-\varphi(t,x'',u_1'',i)|\leq \lambda \left(|x'-x''|+|u_1'-u_1''|\right)$, $\forall x',x''\in \mathbb{R}^n$, $u_1',u_1''\in \mathbb{R}^{m_1}$, $t\in \mathbb{Z}_+$, $i\in \mathfrak N$.
\item [$\gamma$)] The solutions of the closed-loop system:
\begin{align}\label{e5.12'}
x(t+1)&=\left[A_0(t,\theta_t)x(t)+B_{01}(t,\theta_t)u_1(t)+B_{02}(t,\theta_t)\varphi(t,x(t),u_1(t),\theta_t)\right]dt\notag\\
&+\sum_{k=1}^rw_k(t)\left[A_k(t,\theta_t)x(t)+B_{k1}(t,\theta_t)u_1(t)+B_{k2}(t,\theta_t)\varphi(t,x(t),u_1(t),\theta_t)\right]
\end{align}
lie in $\ell_{\tilde{\mathcal{H}}}^2\{0,\infty;\mathbb{R}^{n}\}$ and additionally $\underset{t\rightarrow \infty}{\lim}{\mathbb E}\left[|x_{u_1}(t)|^2\right]=0$, $\forall u_1(\cdot)\in\ell_{\tilde{\mathcal{H}}}^2\{0,\infty;\mathbb{R}^{m_1}\}$, $x_0\in\mathbb{R}^n$.
\end{itemize}
We denote $\tilde{\tilde{\mathcal{U}}}_{ad}$ the set of the strategies of the form $(u_1(\cdot),u_2(\cdot,u_1(\cdot)))$ where $u_1(\cdot) \in \ell_{\tilde{\mathcal{H}}}^2\{0,\infty;\mathbb{R}^{m_1}\}$ and $u_2(\cdot,u_1(\cdot))$ is of type (\ref{e5.12}).

\begin{definition} We say that $(\tilde{\tilde{u}}_1(\cdot),\tilde{\tilde{u}}_2(\cdot,\tilde{\tilde{u}}_1(\cdot)))\in \tilde{\tilde{\mathcal{U}}}_{ad}$ is an equilibrium strategy with respect to the class of admissible strategies $\tilde{\tilde{\mathcal{U}}}_{ad}$ if:
\begin{equation}\label{e5.3"}
\mathcal{J}(x_0,u_1(\cdot),\tilde{\tilde{u}}_2(\cdot,u_1(\cdot)))\leq\mathcal{J}(x_0,\tilde{\tilde{u}}_1(\cdot),\tilde{\tilde{u}}_2(\cdot,\tilde{\tilde{u}}_1(\cdot)))\leq \mathcal{J}(x_0,\tilde{\tilde{u}}_1(\cdot),u_2(\cdot,\tilde{\tilde{u}}_1(\cdot)))
\end{equation}
$\forall u_1\in\ell_{\tilde{\mathcal{H}}}^2\{0,\infty;\mathbb{R}^{m}_1\}$ and $u_2(\cdot,u_1(\cdot))$ with the property that $(u_1(\cdot),\tilde{\tilde{u}}_2(\cdot,u_1(\cdot)))$ and $(\tilde{\tilde{u}}_1(\cdot),u_2(\cdot,\tilde{\tilde{u}}_1(\cdot)))$ are in ${\tilde{\tilde{\cal U}}}_{ad}$.
\end{definition}

\noindent One can see that the strategies of the form:
\begin{equation}\label{e5.13}
u_2(t,u_1(\cdot))=K(t,\theta_t)x(t)+W(t,\theta_t)u_1(t)
\end{equation}
are of type (\ref{e5.12}) if $t\rightarrow K(t,i):\mathbb{Z}_+\rightarrow \mathbb{R}^{m_2\times n}$, $t\rightarrow W(t,i):\mathbb{Z}_+\rightarrow \mathbb{R}^{m_2\times m_1}$ are bounded matrix valued sequences with the property that the linear system of type (\ref{sys_closed}) is ESMS.\\

\noindent Following similar arguments as in the proof of Lemma 5.2 one can states the following Lemma which will be used in the derivation of the main result of this section.

\begin{lem} Assume that the assumptions in Lemma 5.2 are fulfilled. Let $\tilde{\mathbb{X}}(\cdot):\mathbb{Z}_+\rightarrow \mathcal{S}_n^N$ be the bounded and stabilizing solution of the GDTRE (\ref{GRDE}) and $\tilde{F}(t,i)$ be the corresponding stabilizing feedback gain computed via (\ref{e2.8}). We set:
\begin{equation}\label{e5.14}
\tilde{\tilde{u}}_k(t)=\tilde{F}_k(t,\theta_t)\tilde{x}(t),\; k=1,2
\end{equation}
$\tilde{x}(t)$, $t\geq t_0$ being the solution of the closed-loop system 
\begin{align}\label{e5.closed}
x(t+1)&=\left[A_0(t,\theta_t)+B_{0}(t,\theta_t)\tilde{F}(t,\theta_t)\right]x(t)\notag\\
&+\sum_{k=1}^rw_k(t)\left[A_k(t,\eta_t)+B_{k}(t,\eta_t)\tilde{F}(t,\theta_t)\right]x(t)
\end{align}
satisfying $\tilde{x}(t_0)=x_0\in \mathbb{R}^n$.
Under these conditions the following hold:
\begin{itemize}
\item [i)] $\mathcal{J}(x_0,u_1(\cdot),u_2(\cdot))=x_0^T{\mathbb E} \left[\tilde{X}(t_0,\theta_0)\right]x_0+\mathcal{J}(0,u_1(\cdot)-\tilde{u}_1(\cdot),u_2(\cdot)-\tilde{u}_2(\cdot))$, $\forall u_k\in \ell_{\tilde{\mathcal{H}}}^2\{0,\infty;\mathbb{R}^{m_k}\}$ with the property that $(u_1(\cdot),\; u_2(\cdot))\in \mathcal{U}_{ad}$.
\item [ii)] $x_0^T{\mathbb E}\left[\tilde{X}(t_0,\theta_0)\right]x_0\leq \mathcal{J}(x_0,\tilde{\tilde{u}}_1(\cdot),u_2(\cdot)),\; \forall \;\; u_2\in \ell_{\tilde{\mathcal{H}}}^2\{0,\infty;\mathbb{R}^{m_2}\}$ with the property that $(\tilde{\tilde{u}}_1(\cdot),u_2(\cdot))\in \mathcal{U}_{ad}$.
\end{itemize}
\end{lem}
\noindent The next result provides an equilibrium strategy of the zero-sum LQ differential game described by the system (\ref{e5.1}), the quadratic functional (\ref{e5.2}) and the class of admissible strategies $\tilde{\tilde{\mathcal{U}}}_{ad}$.

\begin{thm} Assume that the assumptions of Theorem \ref{T5.6} are fulfilled. 
Let $\tilde{\mathbb{X}}(\cdot)$ be the unique bounded and stabilizing solution of the GDTRE (\ref{GRDE}). Consider $\tilde{F}(t,i)$ being the corresponding stabilizing feedback gain computed as in (\ref{e2.8}). We set:
\begin{equation}\label{e5.15}
\begin{cases}
\tilde{K}(t,i)\triangleq \tilde{\mathbb{V}}_{22}^{-1}(t,i)\left(\begin{array}{cc}\tilde{\mathbb{V}}_{21}(t,i) & \tilde{\mathbb{V}}_{22}(t,i)\end{array}\right)\tilde{F}(t,i)\\
\tilde{W}(t,i)\triangleq-\tilde{\mathbb{V}}_{22}^{-1}(t,i)\tilde{\mathbb{V}}_{21}(t,i)
\end{cases}
\end{equation}
where:
\begin{equation}\label{factor}
\begin{cases}
\tilde{\mathbb{V}}_{21}(t,i)=\tilde{\mathbb{V}}_{22}^{-1}(t,i)\left(R_{12}(t,i)+\Pi_{312}(t)[\tilde{\mathbb{X}}(t)](i)\right)^T\\
\tilde{\mathbb{V}}_{22}(t,i)=\left(R_{22}(t,i)+\Pi_{322}(t)[\tilde{\mathbb{X}}(t)](i)\right)^{\frac{1}{2}}
\end{cases}
\end{equation}
Let $\tilde{\tilde{u}}_2(\cdot,u_1(\cdot))$ be defined by:
\begin{equation}\label{e5.16}
\tilde{\tilde{u}}_2(t,u_1(t))=\tilde{K}(t,\theta_t)x_{u_1}(t)+\tilde{W}(t,\theta_t)u_1(t)
\end{equation}
$x_{u_1}(t)$ being the solution of the initial value problem of type (\ref{equa3}) written for $K(t,\theta_t)$ replaced by $\tilde{K}(t,\theta_t)$ and $W(t,\theta_t)$ replaced by $\tilde{W}(t,\theta_t)$. Under these conditions the following hold:
\begin{itemize}
\item [i)] for each $u_1(\cdot)\in \ell_{\tilde{\mathcal{H}}}^2\{0,\infty;\mathbb{R}^{m_1}\}$, $(u_1(\cdot),\; \tilde{\tilde{u}}_2(\cdot,u_1(\cdot)))$ lies in $\tilde{\tilde{\mathcal{U}}}_{ad}$;
\item[ii)] if $\tilde{\tilde{u}}_1(\cdot)$ is given by (\ref{e5.14}) for $k=1$, then $(\tilde{\tilde{u}}_1(\cdot),\;\tilde{\tilde{u}}_2(\cdot,\tilde{\tilde{u}}_1(\cdot)))$ is an equilibrium strategy of the zero-sum LQ difference game described by (\ref{e5.1}), (\ref{e5.2}) and the class of admissible strategies $\tilde{\tilde{\mathcal{U}}}_{ad}$.
\end{itemize}
\end{thm}

\begin{proof}
\begin{itemize}
\item [i)] By plugging (\ref{e5.16}) in (\ref{e5.1}) one obtains:
\begin{align}\label{e5.17}
x(t+1)&=\left[(A_0(t,\theta_t)+B_{02}(t,\theta_t)\tilde{K}(t,\theta_t))x(t)+(B_{01}(t,\theta_t)+B_{02}(t,\theta_t)\tilde{W}(t,\theta_t))u_1(t)\right]\notag\\
&+\sum_{k=1}^rw_k(t)\left[(A_k(t,\theta_t)+B_{k2}(t,\theta_t)\tilde{K}(t,\theta_t))x(t)+(B_{k1}(t,\theta_t)+B_{k2}(t,\theta_t)\tilde{W}(t,\theta_t))u_1(t)\right]
\end{align}
Hence, it follows from Lemma 5.7 that the autonomous system associated to (\ref{e5.17}) is ESMS.
Following similar steps as in the proof of Theorem 3.2 from \cite{Dragan:14} in the case of the system (\ref{e5.17}), we deduce that for each $u_1(\cdot) \in  \ell_{\tilde{\mathcal{H}}}^2\{0,\infty;\mathbb{R}^{m_1}\}$, the solution $x_{u_1}(\cdot)$ of the system (\ref{e5.17}) are in $ \ell_{\tilde{\mathcal{H}}}^2\{0,\infty;\mathbb{R}^{n}\}$ and $\underset{t\rightarrow\infty}{\lim}{\mathbb E}\left[|x_{u_1}(t)|^2\right]=0$, $\forall i\in \mathfrak N$. Therefore $(u_1(\cdot),\tilde{\tilde{u}}_2(\cdot,u_1(\cdot)))\in \tilde{\tilde{\mathcal{U}}}_{ad}$ for all $u_1(\cdot) \in \ell_{\tilde{\mathcal{H}}}^2\{0,\infty;\mathbb{R}^{m_1}\}$, that means that assertion i) is proved.
\item [ii)] Under the considered assumptions, it follows from the definition of the stabilizing feedback gains $\tilde{F}(\cdot,\cdot)$ that the system (\ref{e5.closed}) is ESMS. Hence, if $\tilde{\tilde{u}}_1(t)=\tilde{F}_1(t,\theta_t)\tilde{x}(t)$, then $\tilde{\tilde{u}}_1(\cdot) \in \ell_{\tilde{\mathcal{H}}}^2\{0,\infty;\mathbb{R}^{m_1}\}$. 
    From (\ref{e5.15}) and (\ref{e5.16}) one obtains that $\tilde{\tilde{u}}_2(t,\tilde{\tilde{u}}_1(t))=\tilde{F}_2(t,\theta_t)\tilde{x}(t),\;\forall t\in \mathbb{Z}_+$. 
    Also, from (\ref{equa_App_B_1}) together with the uniqueness of the solution of an initial value problem, we obtain that:
\begin{equation}\label{e5.18}
\mathcal{J}(x_0,\tilde{\tilde{u}}_1(\cdot),\tilde{\tilde{u}}_2(\cdot))=x_0^T{\mathbb E}[\tilde{X}(t_0,\theta_{t_0})]x_0
\end{equation}
$\forall x_0\in \mathbb{R}^n$. On the other hand, employing the factorization (\ref{factor}) and:
\begin{align}\label{e.72a}
\begin{cases}
&\tilde{\mathbb V}_{11}(t, i)=\left(-{\mathbb R}_{22}^{\sharp}(t,\tilde{X}(t,i)\right)^{\frac{1}{2}}\\
&{\mathbb R}_{22}^{\sharp}(t,\tilde{X}(t,i))=\left(R_{11}(t,i)+\Pi_{311}(t)[\tilde{\mathbb{X}}(t+1)](i)\right)-\tilde{\mathbb{V}}_{21}^T(t,i)\times \tilde{\mathbb{V}}_{21}(t,i)
\end{cases}
\end{align}
one can rewrite (\ref{equa_App_B_1}) in the form:
\begin{align}\label{e5.19}
&\mathcal{J}(x_0,u_1(\cdot),u_2(\cdot))=x_0^T{\mathbb E}[\tilde{X}(t_0,\theta_{t_0})]x_0\notag\\
&+{\mathbb E}\big[\sum_{t=t_0}^\infty\big(|\tilde{\mathbb{V}}_{21}(t,\theta_t)(u_1(t)-F_1(t,\theta_t)x_u(t))+\tilde{\mathbb{V}}_{22}(t,\theta_t)(u_2(t)-F_2(t,\theta_t)x_u(t))|^2\notag\\
&-|\tilde{\mathbb{V}}_{11}(t,\theta_t)(u_1(t)-F_1(t,\theta_t)x_u(t))|^2\big)\big].
\end{align}
Since in the case when $\tilde{\tilde{u}}_2(t)=u_2(t,u_1(t))$, $t\in \mathbb{Z}_+$ the solution $x_u(t)$ of (\ref{e5.1}) coincides with the solution of $x_{u_1}(t)$ of the system (\ref{e5.17}), we obtain from (\ref{e5.19}) that:
\begin{equation}\label{e5.20}
\mathcal{J}(x_0,u_1(\cdot),\tilde{\tilde{u}}_2(\cdot,u_1(\cdot)))=x_0^T{\mathbb E}[\tilde{X}(t_0,\theta_{t_0})]x_0- {\mathbb E}\big[\sum_{t=t_0}^\infty|\tilde{\mathbb{V}}_{11}(t,\theta_t)(u_1(t)-F_1(t,\theta_t)x_u(t))|^2\big]
\end{equation}
for all $u_1(\cdot) \in \ell_{\tilde{\mathcal{H}}}^2\{0,\infty;\mathbb{R}^{m_1}\}$. From (\ref{e5.18}) and (\ref{e5.20}) together with $\tilde{\tilde{u}}_2(t)=\tilde{\tilde{u}}_2(t,\tilde{\tilde{u}}_1(t))$ we obtain that: $\mathcal{J}(x_0,u_1(\cdot),\tilde{\tilde{u}}_2(\cdot,u_1(\cdot)))\leq \mathcal{J}(x_0,\tilde{\tilde{u}}_1(\cdot),\tilde{\tilde{u}}_2(\cdot,\tilde{\tilde{u}}_1(\cdot)))$, for all $u_1(\cdot) \in \ell_{\tilde{\mathcal{H}}}^2\{0,\infty;\mathbb{R}^{m_1}\}$ which confirms the validity of the first inequality from (\ref{e5.3"}).\\
\noindent On the other hand, from Lemma 6.2 ii) we obtain that\\
$\mathcal{J}(x_0,\tilde{\tilde{u}}_1(\cdot),\tilde{\tilde{u}}_2(\cdot,\tilde{\tilde{u}}_1(\cdot)))\leq \mathcal{J}(x_0,\tilde{\tilde{u}}_1(\cdot),u_2(\cdot,\tilde{\tilde{u}}_1(\cdot)))$ $\forall u_2(\cdot,\tilde{\tilde{u}}_1(\cdot))=\varphi(t,x(t),\tilde{\tilde{u}}_1(t),\theta_t)$ which confirms the validity of the second inequality from (\ref{e5.3"}). Thus the proof is complete. \hfill$\blacksquare$
\end{itemize}

\end{proof}
\section*{Appendix A.}
\noindent \textbf{Lemma A.1.}
If $\Gamma(t,i) = \left(
\begin{array}{c}
\Gamma_1(t,i) \\ \Gamma_2(t,i)
\end{array} \right) , \ t\in \mathbb{Z}_+ $ are gain matrices such that $\Gamma_j(t,i) \in \mathbb{R}^{m_j\times n}, j=1,2$, then the Riccati difference equation (\ref{GRDE}) solved by $\mathbb{X}(\cdot)$, may be rewritten in the form:
\begin{equation*}
X(t,i)=\mathcal{L}_\Gamma^*(t)[\mathbb{X}(t+1)](i)+\mathcal{M}_\Gamma(t,i)-\left(\Gamma(t,i)-F(t,i)\right)^T\left(R(t,i)+\Pi_3(t)[\mathbb{X}(t)](i)\right)\star
\end{equation*}
where:
\begin{equation}\label{eq.4.4bis}
\mathcal{M}_\Gamma(t,i)=\left(\begin{array}{cc}{\mathbb I}_n & \Gamma^T(t,i)\end{array}\right)\mathcal{M}(t,i)\star
\end{equation}
\begin{equation*}
F(t,i)=-\left(R(t,i)+\Pi_3(t)[\mathbb{X}(t)](i)\right)^{-1}\left(\Pi_2^T(t)[\mathbb{X}(t)](i)+L^T(t,i)\right)
\end{equation*}
and $\mathcal{L}_\Gamma^*(t)$ is the adjoint of the Lyapunov type operator defined by:
\begin{align}\label{eq.A.1.}
\mathcal{L}_\Gamma(t)[\mathbb{X}](i)&=\sum_{k=0}^r\sum_{j=1}^Np_t(j,i)\big(A_k(t,i)+B_k(t,i)\Gamma(t,i)\big)X(i)\big(A_k(t,i)+B_k(t,i)\Gamma(t,i)\big)^T,\quad1\leq i\leq N
\end{align}
\noindent The proof can be done by direct calculation. The details are omitted.

\section*{Appendix B.}
\noindent In this appendix, we will state and prove several preliminary results that relate the set $\mathcal{A}^{\Sigma}$ to the GDTRE (\ref{Ric_closed}) and the functional (\ref{cost_closed}) respectively. \\

\noindent \textbf{Proposition B.1.} Under the assumptions {\bf H1)} a) and {\bf H4)}, if $\mathcal{A}^{\Sigma}$ is not empty then for any $\left(\mathbb{K}(\cdot), \mathbb{W}(\cdot)\right)\in \mathcal{A}_{\Sigma}$ the unique bounded and stabilizing solution of the corresponding GDTRE (\ref{Ric_closed}) $\tilde{\mathbb{X}}_{\mathbb{KW}}(\cdot)$ verifies the sign conditions (\ref{e2.3}) and (\ref{e2.4}).\\

\noindent \textbf{Proof.} Let
$$ \left(
\begin{array}{cc}
\mathcal{G}_{11}(t,i) &\mathcal{ G}_{12}(t,i) \\
\mathcal{G}_{12}^T(t,i) & \mathcal{G}_{22}(t,i)
\end{array} \right) $$
be the partition of the matrix $R(t,i)+\Pi_3(t)[\tilde{\mathbb{X}}_{\mathbb{K}\mathbb{W}}(t+1)](i)$ such that $\mathcal{G}_{jj}(t,i)\in {\mathcal S}_{m_j}\,, \ j=1,2.$ With these notations, we have that
$\tilde{\mathbb{X}}_{\mathbb{KW}}(\cdot)$ verifies the sign conditions (\ref{e2.3}) and (\ref{e2.4})  if and only if $\mathcal{G}_{22}(t,i) \geq \delta_2 I_{m_2}$ and:
\begin{equation}\label{Eq_3.44}
\mathcal{G}_{11}(t) - \mathcal{G}_{12}(t,i)\,\mathcal{G}_{22}^{-1}(t,i)\,\mathcal{G}_{12}^T(t,i) \leq -\delta_1 I_{m_1} \,.
\end{equation}
First, it follows from Remark \ref{R4.1} ii) that for any $\left(\mathbb{K}(\cdot), \mathbb{W}(\cdot)\right)\in \mathcal{A}^{\Sigma}$ the unique bounded and stabilizing solution of the corresponding GDTRE (\ref{Ric_closed}) satisfies the sign condition $\mathcal{G}_{22}(t,i) \geq \delta_2 I_{m_2}$.\\
\noindent Furthermore, the property (\ref{sign_closed}) leads to:
\begin{align}\label{Eq_3.45}
&R_{\mathbb{W}}(t,i)+\Pi_{\mathbb{W}}(t)[\tilde{\mathbb{X}}_{\mathbb{K}\mathbb{W}}(t+1)](i)  \notag\\
& = \left[R_{12}(t,i) + \Pi_{312}(t)[\tilde{\mathbb{X}}_{\mathbb{K}\mathbb{W}}(t+1)](i)+W^T(t,i)\left(R_{22}(t,i)+\Pi_{322}(t)[\tilde{\mathbb{X}}_{\mathbb{K}\mathbb{W}}(t+1)](i)\right) \right] \notag\\
& \times \left[R_{22}(t,i) + \Pi_{322}(t)[\tilde{\mathbb{X}}_{\mathbb{K}\mathbb{W}}(t+1)](i) \right]^{-1}\star \leq -\varepsilon^2I.
\end{align}
We set:
\begin{equation}\label{equa1_Prop1}
\left(
\begin{array}{cc}
\mathcal{E}_{11}(t,i) & \mathcal{E}_{12}(t,i) \\
\mathcal{E}_{12}^T(t,i) & \mathcal{E}_{22}(t,i)
\end{array} \right)
=
\left(
\begin{array}{cc}
I_{m_1} & W^T(t,i) \\
O & I_{m_2}
\end{array} \right) \,
\left(
\begin{array}{cc}
\mathcal{G}_{11}(t,i) & \mathcal{G}_{12}(t,i) \\
\mathcal{G}_{12}^T(t,i) & \mathcal{G}_{22}(t,i)
\end{array} \right) \,
\star
\end{equation}
One sees that (\ref{Eq_3.45}) may be rewritten as:
\begin{equation}\label{Eq_3.46}
\mathcal{E}_{11}(t,i) - \mathcal{E}_{12}(t,i)\, \mathcal{E}_{22}^{-1}(t,i)\,\mathcal{E}_{12}(t,i) \leq -\varepsilon^2I
\end{equation}
for all $t\in \mathbb{Z}_+$.
Using (\ref{equa1_Prop1}), one can show by direct calculations that:
\begin{equation}\label{equa2_Prop1}
\mathcal{E}_{11}(t,i) - \mathcal{E}_{12}(t,i)\, \mathcal{E}_{22}^{-1}(t,i)\,\mathcal{E}_{12}(t,i)=\mathcal{G}_{11}(t,i) - \mathcal{G}_{12}(t,i)\,\mathcal{G}_{22}^{-1}(t,i)\,\mathcal{G}_{12}^T(t,i)
\end{equation}
\noindent Hence the proof is complete.\hfill$\blacksquare$\\

\noindent \textbf{Proposition B.2.} Under the assumptions {\bf H1)} a), {\bf H2)} and {\bf H3)} a), for each $\left(\mathbb{K}(\cdot), \mathbb{W}(\cdot)\right)\in \mathcal{A}^{\Sigma}$ there exists $\gamma=\gamma(\mathbb{K},\mathbb{W})>0$ such that
$\mathcal{J}_{\mathbb{KW}}(0,0,u_1)\leq -\gamma {\mathbb E}\left[\sum_{t=0}^\infty |u_1(t)|^2\right]$ for all $u_1\in \ell_{\tilde{\mathcal{H}}}^2\{0,\infty;\mathbb{R}^{m_1}\}$.\\

\noindent \textbf{Proof.}  Let $\left(\mathbb{K}(\cdot), \mathbb{W}(\cdot)\right)$ be arbitrary but fixed and let $\tilde{\mathbb{X}}_{\mathbb{KW}}(\cdot)$ be the unique bounded and stabilizing solution of (\ref{Ric_closed}).\\
\noindent Using conditional expectation properties and equation (\ref{Ric_closed}) one gets:
\begin{align}\label{Proposition3.2.equa1}
&\mathcal{J}_{\mathbb{KW}}(t_0,x_0,u_1)=x_0^T\tilde{X}_{\mathbb{K}\mathbb{W}}(t_0,i)x_0\notag\\
&+\sum_{t=t_0}^{\tau} {\mathbb E}\left[\left(\begin{array}{cc}x_{u_1}^T(t) & u_1^T(t)\end{array}\right)\left(\begin{array}{cc}M_{\mathbb{K}\mathbb{W}}^T(t,\theta_t)S_{\mathbb{K}\mathbb{W}}^{-1}(t,\theta_t)\star & M_{\mathbb{K}\mathbb{W}}^T(t,\theta_t) \\\star & S_{\mathbb{K}\mathbb{W}}(t,\theta_t)\end{array}\right)\star\right]
\end{align}
where:
\begin{equation}
\begin{cases}
S_{\mathbb{K}\mathbb{W}}(t,\theta_t)=R_\mathbb{W}(t,\theta_t)+\Pi_\mathbb{W}(t)[{\tilde{\mathbb{X}}}_{\mathbb{K}\mathbb{W}}(t+1)](\theta_t)\\
M_{\mathbb{K}\mathbb{W}}^T(t,\theta_t)=\Pi_{\mathbb{K}\mathbb{W}}(t)[\tilde{{\mathbb{X}}}_{\mathbb{K}\mathbb{W}}(t+1)](\theta_t)+L_\mathbb{KW}(t,\theta_t).
\end{cases}
\end{equation}
By direct calculation, one can rewrite (\ref{Proposition3.2.equa1}) as:
\begin{align}
\label{Proposition3.2.equa2}
&\mathcal{J}_{\mathbb{KW}}(t_0,x_0,u_1)=x_0^T\tilde{X}_{\mathbb{K}\mathbb{W}}(t_0,i)x_0+\notag\\
&+\sum_{t=t_0}^\tau {\mathbb E}
\left[\left(u_1^T(t)+x_{u_1}^T(t)M_{\mathbb{K}\mathbb{W}}^T(t,\theta_t)S_{\mathbb{K}\mathbb{W}}^{-1}(t,\theta_t)\right)S_{\mathbb{K}\mathbb{W}}(t,\theta_t)\star\right]
\end{align}
It follows from the equality above and (\ref{sign_closed}) that:
\begin{equation}
\mathcal{J}_{\mathbb{KW}}(0,0,u_1)\leq 0
\end{equation}
for every $u_1\in \ell_{\tilde{\mathcal{H}}}^2\{0,\infty;\mathbb{R}^{m_1}\}$.
We will now show that:
\begin{eqnarray}\label{e3.10c}
\sup\{\mathcal{J}_{\mathbb{KW}}(0,0,u_1); u_1\in \ell_{\tilde{\mathcal{H}}}^2\{0,\infty;\mathbb{R}^{m_1}\}, \|u_1\|=1\}<0.
\end{eqnarray}
To this end, we will use a proof by contradiction. Assume that (\ref{e3.10c}) is not true. This means that there exists a sequence $\{u_1^j\}_{j\geq 0}\subset \ell_{\tilde{\mathcal{H}}}^2\{0,\infty;\mathbb{R}^{m_1}\}$ with $\|u_1^j\| =1$ and $\underset{j\to \infty}{\lim}\mathcal{J}_{\mathbb{KW}}(0,0, u_1^j)=0$.
It follows from (\ref{Proposition3.2.equa2}) and (\ref{sign_closed}) that:
\begin{eqnarray}\label{e3.10d}
\lim_{j\to \infty} {\mathbb E}\left[\sum_{t=0}^{\infty} |f_j(t)|^2\right]=0
\end{eqnarray}
where $f_j(t)=u_1^j(t)-\tilde F_{\mathbb{KW}}(t,\theta_t) x_{u_1^j}(t)$ with $x_{u_1^j}(t)$ being the solution of (\ref{equa3}) determined by the input $u_1^j$ and having the initial value $x_j(0)=0$ and:
\begin{equation}
\label{Stab_gain}
F_{\mathbb{KW}}(t,i)=-S_{\mathbb{K}\mathbb{W}}^{-1}(t,i)M_{\mathbb{K}\mathbb{W}}(t,i).
\end{equation}
\noindent Now, we rewrite the equation (\ref{equa3}) satisfied by $x_{u_1^j}(\cdot)$ as:
\begin{align}\label{e3.10e}
&x_{u_1^j}(t+1)=\left(A_{0K}(t,\theta_t)+B_{0W}(t, \theta_t)\tilde F_{\mathbb{KW}}(t,\eta_t)\right)x_{u_1^j}(t)+B_{0W}(t, \theta_t)f_j(t)\notag\\
&+\sum_{k=1}^r w_k(t)\left[\left(A_{kK}(t, \theta_t)+B_{kW}(t, \theta_t)\tilde F_{\mathbb{KW}}(t,\eta_t)\right)x_{u_1^j}(t)+B_{kW}(t, \theta_t)f_j(t)\right].
\end{align}
By taking into account that $\tilde {\mathbb{F}}_{\mathbb{KW}}(\cdot)$ is the stabilizing feedback gain, we obtain from Corollary 3.9 in \cite{carte2010} and (\ref{e3.10d}) that
\begin{eqnarray}\label{e3.10f}
\lim_{j\to \infty} {\mathbb E}\left[\sum_{t=0}^{\infty} |x_{u_1^j}(t)|^2\right]=0
\end{eqnarray}
Combining (\ref{e3.10d}) and (\ref{e3.10f}) we deduce that
$$\lim_{j\to \infty} {\mathbb E}\left[\sum_{t=0}^{\infty} |u_1^j(t)|^2\right]=0$$
This  contradicts the fact that $\|u_1^j\|=1$. Hence, (\ref{e3.10c}) is true. Finally, the conclusion is obtained from (\ref{e3.10c}) using the equality
$$\mathcal{J}_{\mathbb{KW}}\left(0,0,\frac{u_1}{\|u_1\|}\right)= \frac{1}{\|u_1\|^2}\mathcal{J}_{\mathbb{KW}}(0,0,u_1)$$
for all $0\neq u_1\in \ell_{\tilde{\mathcal{H}}}^2\{0,\infty;\mathbb{R}^{m_1}\}$.\hfill$\blacksquare$\\

\noindent For each $\left(\mathbb{K}(\cdot), \mathbb{W}(\cdot)\right) \in \mathcal{A}^{\Sigma}$ and $\tau\in \mathbb{Z}_+$ we consider $\mathbb{X}_{\mathbb{K}\mathbb{W}}^\tau(t)=\left(X_{\mathbb{K}\mathbb{W}}^\tau(t,1),  \cdots , X_{\mathbb{K}\mathbb{W}}^\tau(t,N)\right)$ the solution of (\ref{Ric_closed}) satisfying the condition $X_{\mathbb{K}\mathbb{W}}^\tau(\tau+1,i)=0$, $1\leq i\leq N$ and $\mathcal{I}_{\mathbb{K}\mathbb{W}}(\tau)\subset\left[0,  \tau+1\right]$ is the maximal interval on which $\mathbb{X}_{\mathbb{K}\mathbb{W}}^\tau(\cdot)$ is defined.\\

\noindent \textbf{Proposition B.3.}
For all $(\mathbb{K}(\cdot),\mathbb{W}(\cdot))\in \mathcal{A}^{\Sigma}$, the following hold:
\begin{itemize}
\item [i)] $\mathcal{I}_{\mathbb{K}\mathbb{W}}(\tau)=\left[0,  \tau+1\right]$, for all $\tau \in \mathbb{Z}_+$.
\item [ii)] $0\leq X_{\mathbb{K}\mathbb{W}}^{\tau}(t,i)\leq \tilde{X}_{\mathbb{K}\mathbb{W}}(t,i)$,
for all $0\leq t\leq \tau+1$, $i\in \mathfrak{N}$.
\item [iii)] $R_\mathbb{W}(t,i)+\Pi_\mathbb{W}(t)[{\mathbb{X}}_{\mathbb{K}\mathbb{W}}^\tau(t+1)](i)\leq -\xi I_{m_1}$, for some positive scalar $\xi$, $(t,i)\in \left[0, \tau+1\right]\times \mathfrak{N},\;\tau\in \mathbb{Z}_+$.
\end{itemize}

\noindent \textbf{Proof.} First, using Lemma A.1 (in Appendix A) applied to the Riccati equation (\ref{Ric_closed}) with ${\Gamma}(t,i)={F}_{\mathbb{KW}}^\tau(t,i)$, $t\in \mathcal{I}_{\mathbb{K}\mathbb{W}}(\tau),\; i\in \mathfrak{N}$, where ${F}_{\mathbb{KW}}^\tau(t,i)$ is computed as in (\ref{Stab_gain}) with $\tilde{X}_{\mathbb{KW}}(t,i)$ replaced by ${X}_{\mathbb{KW}}^\tau(t,i)$, one shows that $t\longrightarrow \tilde{\mathbb{X}}_{\mathbb{KW}}(t)-\mathbb{X}_{\mathbb{KW}}^\tau(t):\mathcal{I}_{\mathbb{K}\mathbb{W}}(\tau)\longrightarrow \mathcal{S}_n^N$ is a solution of the Lyapunov type differential equation:
\begin{equation}\label{Prop3.3_eq.1}
\tilde{\mathbb{X}}_{\mathbb{KW}}(t)-\mathbb{X}_{\mathbb{KW}}^\tau(t)=
\left(\mathcal{L}_\Gamma^{\mathbb{KW}}(t)\right)^*\left[\tilde{\mathbb{X}}_{\mathbb{KW}}(t+1)-\mathbb{X}_{\mathbb{KW}}^\tau(t+1)\right]+\mathbb{H}^\tau(t),\; t\in \mathcal{I}_{\mathbb{K}\mathbb{W}}(\tau)
\end{equation}
where $\left(\mathcal{L}_\Gamma^{\mathbb{KW}}(t)\right)^*$ is the adjoint operator of the Lyapunov type operator defined by:
\begin{align}\label{eq.A.1.}
\mathcal{L}_\Gamma^{\mathbb{KW}}(t)[\mathbb{X}](i)&=\sum_{k=0}^r\sum_{j=1}^Np_t(j,i)\big(A_{kK}(t,i)+B_{kW}(t,i)\Gamma(t,i)\big)X(i)\big(A_{kK}(t,i)+B_{kW}(t,i)\Gamma(t,i)\big)^T,
\end{align}
$1\leq i\leq N$, and:
\begin{equation}\label{Prop3.3_eq.2}
\mathbb{H}^\tau(t,i)=-\left({F}_{\mathbb{KW}}^\tau(t,i)-\tilde{{F}}_{\mathbb{KW}}(t,i)\right)\left(R_\mathbb{W}(t,i)+\Pi_\mathbb{W}(t)[{\tilde{\mathbb{X}}}_{\mathbb{K}\mathbb{W}}(t)](i)\right)\star,\; t\in \mathcal{I}_{\mathbb{K}\mathbb{W}}(\tau)
\end{equation}
It follows from (\ref{sign_closed}) that $\mathbb{H}^\tau(t,i)\geq 0$ for all $t\in \mathcal{I}_{\mathbb{K}\mathbb{W}}(\tau)$, $i\in \mathfrak{N}$.

\noindent Hence using the fact that $\tilde{X}_{\mathbb{KW}}(\tau+1,i)-X_{\mathbb{KW}}^\tau(\tau+1,i)=\tilde{X}_{\mathbb{KW}}(\tau+1,i)\geq0$
one can shows by induction that:
\begin{equation}\label{Prop3.3_eq.3}
X_{\mathbb{KW}}^\tau(\tau,i)\leq \tilde{X}_{\mathbb{KW}}(\tau,i)
\end{equation}
for all $t\in \mathcal{I}_{\mathbb{K}\mathbb{W}}(\tau)$, $i\in \mathfrak{N}$.
Now, if one interprets equation (\ref{Ric_closed}) satisfied by $\mathbb{X}_{\mathbb{KW}}^\tau(\cdot)$ as a generalized Lyapunov equation with nonnegative free term, one obtains:
\begin{equation}\label{Prop3.3_eq.4}
X_{\mathbb{KW}}^\tau(\tau,i)\geq 0
\end{equation}
for all $t\in \mathcal{I}_{\mathbb{K}\mathbb{W}}(\tau)$, $i\in \mathfrak{N}$ (see Chapter 3 in \cite{carte2010} for further details).
Hence i) is proved from (\ref{Prop3.3_eq.3}) and (\ref{Prop3.3_eq.4}). Assertion ii) is a consequence of property i) and (\ref{Prop3.3_eq.3}), (\ref{Prop3.3_eq.4}). Assertion iii) follows from (\ref{sign_closed}), (\ref{Prop3.3_eq.3}) and the property i). The proof is completed.\hfill$\blacksquare$\\

\noindent \textbf{Remark B.1.}  If $\mathbb{K}(\cdot)\in \mathcal{A}_{0}^{\Sigma}$ and $\tau\in \mathbb{Z}_+$ we consider $\mathbb{X}_{\mathbb{K}}^\tau(t)=\left(X_{\mathbb{K}}^\tau(t,1), \cdots, X_{\mathbb{K}}^\tau(t,N)\right)$ the solution of (\ref{Ric_closed}) satisfying the condition $X_{\mathbb{K}}^\tau(\tau+1,i)=0$, $1\leq i\leq N$ and $\mathcal{I}_{\mathbb{K}}(\tau)\subset\left[0,  \tau+1\right]$ the maximal interval on which $\mathbb{X}_{\mathbb{K}}^\tau(\cdot)$ is defined. One can specializes in a straightforward way the results obtained in Proposition B.3 to $\mathbb{X}_{\mathbb{K}}^\tau(t)$.\\

\noindent Let us now consider the following quadratic functional on finite horizon time:
\begin{equation}\label{functional_tau}
\mathcal{J}_{\mathbb{KW}}(t_0,\tau,x_0,i;u_1)={\mathbb E}\left[\sum_{t=t_0}^\tau\left(\begin{array}{c}x_{u_1}(t) \\u_1(t)\end{array}\right)^T\left(\begin{array}{cc}M_\mathbb{K}(t,\theta_t) & L_{\mathbb{KW}}(t,\theta_t) \\\star & R_\mathbb{W}(t,\theta_t)\end{array}\right)\star\Big |\theta_{t_0}=i\right]
\end{equation}
$0\leq t_0<\tau$, $i\in \mathfrak{N}$, $x_0\in \mathbb{R}^n$, $u_1\in \ell_{\tilde{\mathcal{H}}}^2\{t_0,\tau;\mathbb{R}^{m_1}\}$.\\

\noindent \textbf{Proposition B.4.}
Under assumptions {\bf H1)}, {\bf H2)} and {\bf H3)} a), if $\left(\mathbb{K}(\cdot),  \mathbb{W}(\cdot)\right)\in \mathcal{A}^{\Sigma}$, then there exist positive constants $\alpha_0$, $\alpha_1$ such that:
\begin{equation}\label{Prop3.4_eq.1}
\tilde{\mathcal{J}}_{\mathbb{KW}}(t_0,\tau,x_0;u_1)\leq \alpha_0|x_0|^2-\alpha_1{\mathbb E}\left[\sum_{t=t_0}^\tau|u_1(t)|^2\right]
\end{equation}
for all $\tau>0$, $u_1\in \ell_{\tilde{\mathcal{H}}}^2\{t_0,\tau;\mathbb{R}^{m_1}\}$, $x_0\in \mathbb{R}^n$, where:
\begin{equation}\label{Prop3.4_eq.2}
\tilde{\mathcal{J}}_{\mathbb{KW}}(t_0,\tau,x_0;u_1)=\sum_{i=1}^N\pi_{t_0}(i){\mathcal{J}}_{\mathbb{KW}}(t_0,\tau,x_0,i;u_1).
\end{equation}

\noindent \textbf{Proof.} Let ${\mathbb Z}^\tau(t)=(Z^\tau(t,1),  \cdots, Z^\tau(t,N))$  be the solution of the linear difference equation:
\begin{equation}\label{Lyap_closed}
{Z}(t,i)=\Pi_\mathbb{K}(t)[\mathbb{Z}(t+1)](i)+M_\mathbb{K}(t,i)
\end{equation}
satisfying the condition $Z^\tau(\tau+1,i)=0$, $1\leq i\leq N$. Since (\ref{Lyap_closed}) is a linear difference equation it follows that $\mathbb{Z}^\tau(\cdot)$ is well defined for every $t\leq \tau+1$. On the other hand, the system (\ref{sys_closed}) being ESMS-CI, we deduce (under {\bf H1)} b)) that there exists $\gamma>0$ (not depending upon $\tau,\; t,\; i$) such that:
\begin{equation}\label{Lyap_closed'}
0\leq{Z}^\tau(t,i)\leq \gamma I_n
\end{equation}
for all $\tau>0$, $t\leq \tau+1$, $i\in \mathfrak{N}$. Using conditional expectation properties and invoking (\ref{Lyap_closed}) and (\ref{Prop3.4_eq.2}) one gets:
\begin{align}\label{Prop3.4_eq.3}
\tilde{\mathcal{J}}_{\mathbb{KW}}(t_0,\tau,x_0;u_1)&=\sum_{i=1}^N\pi_{t_0}(i)x_0^T{Z}^\tau(t_0,i)x_0+\notag\\
&+\sum_{t=t_0}^\tau {\mathbb E}\left[\left(\begin{array}{cc}x^T(t;x_0,u_1) & u_1^T(t)\end{array}\right)\left(\begin{array}{cc}0 & \mathcal{L}_{\mathbb{K}\mathbb{W}}(t,\theta_t) \\\star & \mathcal{R}_{\mathbb{K}\mathbb{W}}(t,\theta_t)\end{array}\right)\star\right]
\end{align}
where:
\begin{equation}
\begin{cases}
\mathcal{R}_{\mathbb{K}\mathbb{W}}(t,\theta_t)=R_\mathbb{W}(t,\theta_t)+\Pi_\mathbb{W}(t)[{\mathbb{Z}}^\tau(t+1)](\theta_t)\\
\mathcal{L}_{\mathbb{K}\mathbb{W}}(t,\theta_t)=\Pi_{\mathbb{K}\mathbb{W}}(t)[{\mathbb{Z}}^\tau(t+1)](\theta_t)+L_\mathbb{KW}(t,\theta_t)
\end{cases}
\end{equation}
for all $\tau>0$, $u_1\in \ell_{\tilde{\mathcal{H}}}^2\{t_0,\tau;\mathbb{R}^{m_1}\}$, $x_0\in \mathbb{R}^n$ where $x(t;x_0,u_1)$ is the solution of (\ref{equa3}) satisfying the initial condition $x(t_0;x_0,u_1)=x_0$. It follows from the splitting $x(t;x_0,u_1)=x(t;x_0,0)+x(t;0,u_1)$ that:
\begin{align}\label{Prop3.4_eq.4}
\tilde{\mathcal{J}}_{\mathbb{KW}}(t_0,\tau,x_0;u_1)=\sum_{i=1}^N\pi_{t_0}(i)x_0^T{Z}^\tau(t_0,i)x_0+
\tilde{\mathcal{J}}_{\mathbb{KW}}(t_0,\tau,0;u_1)+2{\mathbb E}\left[\sum_{t=t_0}^\tau x(t;x_0,0)\mathcal{L}_{\mathbb{K}\mathbb{W}}(t,\theta_t)u_1(t)\right]
\end{align}
for all $\tau>0$, $u_1\in \ell_{\tilde{\mathcal{H}}}^2\{t_0,\tau;\mathbb{R}^{m_1}\}$, $x_0\in \mathbb{R}^n$.
Using Proposition B.2 one can show that there exists $\rho>0$ such that:
\begin{equation}
\label{Prop3.4_eq.5}
\tilde{\mathcal{J}}_{\mathbb{KW}}(t_0,\tau,0;u_1)\leq-\rho {\mathbb E}\left[\sum_{t=t_0}^\tau|u_1(t)|^2\right]
\end{equation}
Hence we get:
\begin{align}\label{Prop3.4_eq.6}
\tilde{\mathcal{J}}_{\mathbb{KW}}(t_0,\tau,x_0;u_1)\leq \gamma|x_0|^2-\rho {\mathbb E}\left[\sum_{t=t_0}^\tau|u_1(t)|^2\right] +2{\mathbb E}\left[\sum_{t=t_0}^\tau x(t;x_0,0)\mathcal{L}_{\mathbb{K}\mathbb{W}}(t,\theta_t)u_1(t)\right]
\end{align}
for all $\tau>0$, $u_1\in \ell_{\tilde{\mathcal{H}}}^2\{t_0,\tau;\mathbb{R}^{m_1}\}$, $x_0\in \mathbb{R}^n$.
Note also that for every $\lambda>0$ we have:
\begin{align}\label{Prop3.4_eq.7}
2x(t;x_0,0)\mathcal{L}_{\mathbb{K}\mathbb{W}}(t,\theta_t)u_1(t)\leq \lambda|u_1(t)|^2+\lambda^{-1}x^T(t;x_0,0)\mathcal{L}_{\mathbb{K}\mathbb{W}}(t,\theta_t)\mathcal{L}_{\mathbb{K}\mathbb{W}}^T(t,\theta_t)x(t;x_0,0).
\end{align}
Since $t\rightarrow \mathcal{L}_{\mathbb{K}\mathbb{W}}(t,\theta_t)$ is bounded, there exists $\gamma_1>0$ such that $0\leq \mathcal{L}_{\mathbb{K}\mathbb{W}}(t,\theta_t)\mathcal{L}_{\mathbb{K}\mathbb{W}}^T(t,\theta_t)\leq \gamma_1I_n$ for all $t\in \mathbb{Z}_+$. It follows that:
\begin{align}\label{Prop3.4_eq.8}
\tilde{\mathcal{J}}_{\mathbb{KW}}(t_0,\tau,x_0;u_1)\leq \gamma|x_0|^2-(\rho-\lambda) {\mathbb E}\left[\sum_{t=t_0}^\tau|u_1(t)|^2\right] +\gamma_1\lambda^{-1}{\mathbb E}\left[\sum_{t=t_0}^\tau |x(t;x_0,0)|^2\right].
\end{align}
Because system (\ref{sys_closed}) is ESMS-CI, we have that: ${\mathbb E}\left[\sum_{t=t_0}^\tau |x(t;x_0,0)|^2\right]\leq \gamma_2|x_0|^2$ for all $x_0\in \mathbb{R}^n$ and for some $\gamma_2>0$ which does not depend upon $x_0$. Hence it follows from (\ref{Prop3.4_eq.8}) that:
\begin{align}\label{Prop3.4_eq.9}
\tilde{\mathcal{J}}_{\mathbb{KW}}(t_0,\tau,x_0;u_1)\leq (\gamma+\gamma_1\gamma_2\lambda^{-1})|x_0|^2-(\rho-\lambda) {\mathbb E}\left[\sum_{t=t_0}^\tau|u_1(t)|^2\right]
\end{align}
for all $\tau>0$, $u_1\in \ell_{\tilde{\mathcal{H}}}^2\{t_0,\tau;\mathbb{R}^{m_1}\}$, $x_0\in \mathbb{R}^n$.
The proof ends by taking $0<\lambda<\rho$. \hfill$\blacksquare$

\section*{Appendix C.}
In this Appendix, we will state and prove several technical results that will be used in the proof process for the existence conditions of the bounded and stabilizing solution of (\ref{GRDE}). These results rely partly on the material given in Appendix B.\\

\noindent Let:
\begin{equation}\label{equa_F}
\begin{cases}
F_\tau(t,i)=-\left(R(t,i)+\Pi_3(t)[\mathbb{X}_\tau(t+1)](i)\right)^{-1}\left(\Pi_2^T(t)[\mathbb{X}_\tau(t+1)](i)+L^T(t,i)\right)\\
F_1^\tau(t,i)=\left(\begin{array}{cc}{\mathbb I}_{m_1} & \mathbf{0}\end{array}\right)F_\tau(t,i)\\
F_2^\tau(t,i)=\left(\begin{array}{cc}\mathbf{0} & {\mathbb I}_{m_2}\end{array}\right)F_\tau(t,i)
\end{cases}
,t\in \mathcal{I}(\tau),\; 1\leq i\leq N
\end{equation}
Let $t_0\in{\cal I}(\tau), t_0\leq \tau$ be fixed and  $x_{F_\tau}(t,t_0,x_0)$ be the solution of the initial value problem (IVP):
\begin{equation}
\begin{cases}
x(t+1)=\left[A_0(t,\theta_t)+B_0(t,\theta_t)F_\tau(t,\theta_t)+\sum_{k=1}^rw_k(t)\left(A_k(t,\theta_t)+B_k(t,\theta_t)F_\tau(t,\theta_t)\right)\right]x(t)\\
x(t_0)=x_0,\; x_0\in \mathbb{R}^n
\end{cases}
\end{equation}
and:
\begin{equation}
\begin{cases}
{\tilde{u}}_1^\tau(t)=F_1^\tau(t,\theta_t)x_{F_\tau}(t,t_0,x_0)\\
{\tilde{u}}_2^\tau(t)=F_2^\tau(t,\theta_t)x_{F_\tau}(t,t_0,x_0).
\end{cases}
\end{equation}
Obviously ${\tilde{u}}_1^\tau\in\ell_{\tilde{\mathcal{H}}}^2\{t_0,\tau;\mathbb{R}^{m_1}\}$ and ${\tilde{u}}_2^\tau\in\ell_{\tilde{\mathcal{H}}}^2\{t_0,\tau;\mathbb{R}^{m_2}\}$, $t_0\in \mathcal{I}(\tau),\;t_0<\tau$.\\

\noindent If $u_2\in\ell_{\tilde{\mathcal{H}}}^2\{t_0,\tau;\mathbb{R}^{m_2}\}$, $t_0\in \mathcal{I}(\tau),\;t_0<\tau$, $x_{F_1^\tau}(t,t_0,x_0;u_2)$ is the solution of the IVP:
\begin{equation}
\begin{cases}
x(t+1)=\left[A_0(t,\theta_t)+B_{01}(t,\theta_t)F_1^\tau(t,\theta_t)+\sum_{k=1}^rw_k(t)\left(A_k(t,\theta_t)+
B_{k1}(t,\theta_t)F_1^\tau(t,\theta_t)\right)\right]x(t)+\\
\quad\quad\quad+\left[B_{02}(t,\theta_t)+\sum_{k=1}^rw_k(t)B_{k2}(t,\theta_t)\right]u_2(t)\\
x(t_0)=x_0
\end{cases}
\end{equation}
with $u_1^\tau(t,u_2)=F_1^\tau(t,\theta_t)x_{F_1^\tau}(t,t_0,x_0;u_2)$. Obviously $u_1^\tau(\cdot,u_2)\in\ell_{\tilde{\mathcal{H}}}^2\{t_0,\tau;\mathbb{R}^{m_1}\}$.\\

\noindent Consider the cost:
\begin{equation}
\mathcal{J}(t_0,\tau,x_0,i,u_1,u_2)=\sum_{t=t_0}^\tau {\mathbb E}\left[\left(\begin{array}{c}x_{\bar{u}}(t) \\\bar{u}(t)\end{array}\right)^T\left(\begin{array}{cc}M(t,\theta_t) & L(t,\theta_t) \\\star & R(t,\theta_t)\end{array}\right)\star\Big|\theta_{t_0}=i\right]
\end{equation}
$t_0\in \mathcal{I}(\tau)$, $t_0\leq \tau$, $u_1\in\ell_{\tilde{\mathcal{H}}}^2\{t_0,\tau;\mathbb{R}^{m_1}\}$ and $u_2\in\ell_{\tilde{\mathcal{H}}}^2\{t_0,\tau;\mathbb{R}^{m_2}\}$, where $\bar{u}(t)=\left(\begin{array}{cc}u_1^T(t) & u_2^T(t)\end{array}\right)^T$ and $x_{\bar{u}}(t)$ is the solution of (\ref{equa1}) corresponding to $u_1$ and $u_2$ with $x(t_0)=x_0$.\\

\noindent \textbf{Lemma C.1.} Under the assumptions {\bf H1)} a), {\bf H2)}, {\bf H3)} a) and {\bf H4)}, for all $\tau>0$, $t_0\in \mathcal{I}(\tau)$, $t_0\leq \tau$, $x_0\in \mathbb{R}^n$, $i\in \mathfrak{N}$, $u_1\in \ell_{\tilde{\mathcal{H}}}^2\{t_0,\tau;\mathbb{R}^{m_1}\}$ and $u_2\in\ell_{\tilde{\mathcal{H}}}^2\{t_0,\tau;\mathbb{R}^{m_2}\}$ we have:
\begin{itemize}
\item [i)] $\mathcal{J}(t_0,\tau,x_0,i,u_1,u_2)=x_0^TX_\tau(t_0,i)x_0+\mathcal{J}(t_0,\tau,0,i,u_1-\tilde{u}_1^\tau,u_2-\tilde{u}_2^\tau)$
\end{itemize}
\begin{itemize}
\item [ii)]  for all $\tau>0$, $t_0\in \mathcal{I}(\tau)$, $x_0\in \mathbb{R}^n$, $i\in \mathfrak{N}$ and $u_2\in\ell_{\tilde{\mathcal{H}}}^2\{t_0,\tau;\mathbb{R}^{m_2}\}$ we have:
$$x_0^TX_\tau(t_0,i)x_0\leq \mathcal{J}(t_0,\tau,x_0,i,\tilde{u}_1^\tau,u_2).$$
\end{itemize}

\noindent \textbf{Proof.} Let $\hat{x}_{\bar{u}}(t,t_0,0;\bar{u}-\tilde{\bar{u}}^\tau)$ be the solution of the IVP:
\begin{equation}
\begin{cases}
x(t+1)=\left[A_0(t,\theta_t)+\sum_{k=1}^rw_k(t)A_k(t,\theta_t)\right]x(t)+\left[B_0(t,\theta_t)+\sum_{k=1}^rw_k(t)B_k(t,\theta_t)\right]\left(\bar{u}(t)-\tilde{\bar{u}}^\tau(t)\right)\\
x(t_0)=0
\end{cases}
\end{equation}
with ${\tilde{\bar{u}}}^\tau(t)=\left(\begin{array}{cc}(\tilde{u}_1^\tau)^T(t) & (\tilde{u}_2^\tau)^T(t)\end{array}\right)^T$. Since $x_{\bar{u}}(t)=\hat{x}_{\bar{u}}(t)+x_{F_\tau}(t)$ and using conditional expectation properties together with relation (\ref{GRDE}), one obtains:
\begin{align}\label{Lemma3.1.equa1}
&\mathcal{J}(t_0,\tau,x_0,i,u_1,u_2)=x_0^TX_\tau(t_0,i)x_0\notag\\
&+\sum_{t=t_0}^\tau {\mathbb E}\left[[(\hat{x}_{\bar{u}}(t)+x_{F_\tau}(t))^T \quad  (u_1^T(t)\quad u_2^T(t))]\left(\begin{array}{cc}U_{\tau}^T(t,\theta_t)S_{\tau}^{-1}(t,\theta_t)\star & U_{\tau}^T(t,\theta_t) \\\star & S_{\tau}(t,\theta_t)\end{array}\right)\star\Big|\theta_{t_0}=i\right]
\end{align}
where:
\begin{equation}
\begin{cases}
S_{\tau}(t,\theta_t)=R(t,\theta_t)+\Pi_3(t)\left[\mathbb{X}_{\tau}(t+1)\right](\theta_t)\\
U_{\tau}^T(t,\theta_t)=\Pi_{2}(t)[\mathbb{X}_{\tau}(t+1))](\theta_t)+L(t,\theta_t).
\end{cases}
\end{equation}
By direct calculation, one can shows that (\ref{Lemma3.1.equa1}) can be rewritten as follows:
\begin{align}
&\mathcal{J}(t_0,\tau,x_0,i,u_1,u_2)=x_0^TX_{\tau}(t_0,i)x_0\notag\\
&+\sum_{t=t_0}^\tau {\mathbb E}\left[[\hat{x}_{\bar{u}}^T(t)\quad (\bar{u}(t)-\tilde{\bar{u}}^\tau(t))^T]\left(\begin{array}{cc}U_{\tau}^T(t,\theta_t)S_{\tau}^{-1}(t,\theta_t)\star & U_{\tau}^T(t,\theta_t) \\\star & S_{\tau}(t,\theta_t)\end{array}\right)\star\Big|\theta(t_0)=i\right]\notag\\
&=x_0^TX_{\tau}(t_0,i)x_0+\mathcal{J}(t_0,\tau,0,i,u_1-\tilde{u}_1^\tau,u_2-\tilde{u}_2^\tau)
\end{align}
hence i) is proved.\\
\noindent From i) it follows that $\mathcal{J}(t_0,\tau,x_0,\tilde{u}_1^\tau,u_2)=x_0^TX_\tau(t_0,i)x_0+\mathcal{J}(t_0,\tau,0,i,0,u_2-\tilde{u}_2^\tau)$. Assertion ii) follows directly from the fact that $\mathcal{J}(t_0,\tau,0,i,0,u_2-\tilde{u}_2^\tau)\geq0$ (via the assumption  ${\bf H4)}$).
This completes the proof. \hfill$\blacksquare$\\

\noindent \textbf{Lemma C.2.} Assume:
\begin{itemize}
\item[a)] The assumptions in Lemma C.1. are fulfilled.
\item[b)] ${\cal A}^{\Sigma}$ is not empty.
\end{itemize}
Under these  assumptions, we have:
\begin{equation}
\mathcal{J}(t_0,\tau,x_0,i,u_1,u_2^{\mathbb{K}\mathbb{W}})\leq x_0^TX^\tau_{\mathbb{K}\mathbb{W}}(t_0,i)x_0
\end{equation}
for all $\tau>0$, $t_0\in \mathcal{I}(\tau)$, $x_0\in \mathbb{R}^n$, $i\in \mathfrak{N}$, $u_1\in \ell_{\tilde{\mathcal{H}}}^2\{t_0,\tau;\mathbb{R}^{m_1}\}$ and $(\mathbb{K}(\cdot),\mathbb{W}(\cdot))\in \mathcal{A}^{\Sigma}$.\\

\noindent \textbf{Proof.} Let $x_{\mathbb{K}\mathbb{W}}(t,t_0,x_0;u_1)$ be the solution of the IVP:
\begin{equation}
\begin{cases}
x(t+1)=\left[A_{0{\mathbb K}}(t,\theta_t)+\sum_{k=1}^rw_k(t)A_{k{\mathbb K}}(t,\theta_t)\right]x(t)+\left[B_{0{\mathbb W}}(t,\theta_t)+\sum_{k=1}^rw_k(t)B_{k{\mathbb W}}(t,\theta_t)\right]u_1(t)\\
x(t_0)=x_0.
\end{cases}
\end{equation}
Using conditional expectation properties and equation (\ref{Ric_closed}) one gets:
\begin{align}\label{Lemma3.2.equa1}
&\mathcal{J}(t_0,\tau,x_0,i,u_1,u_2^{\mathbb{K}\mathbb{W}})=x_0^TX_{\mathbb{K}\mathbb{W}}^\tau(t_0,i)x_0\notag\\
&+\sum_{t=t_0}^\tau {\mathbb E}\left[\left[\begin{array}{cc}x_{\mathbb{K}\mathbb{W}}^T(t) & u_1^T(t)\end{array}\right]\left(\begin{array}{cc}(M_{\mathbb{K}\mathbb{W}}^\tau)^T(t,\theta_t)(S_{\mathbb{K}\mathbb{W}}^\tau)^{-1}(t,\theta_t)\star & (M_{\mathbb{K}\mathbb{W}}^\tau)^T(t,\theta_t) \\\star & S_{\mathbb{K}\mathbb{W}}^\tau(t,\theta_t)\end{array}\right)\star\Big|\theta(t_0)=i\right]
\end{align}
where:
\begin{equation}
\begin{cases}
S_{\mathbb{K}\mathbb{W}}^\tau(t,\theta_t)=R_\mathbb{W}(t,\theta_t)+\Pi_\mathbb{W}(t)[{\mathbb{X}}_{\mathbb{K}\mathbb{W}}^\tau(t+1)](\theta_t)\\
(M_{\mathbb{K}\mathbb{W}}^\tau)^T(t,\theta_t)=\Pi_{\mathbb{K}\mathbb{W}}(t)
[{\mathbb{X}}_{\mathbb{K}\mathbb{W}}^\tau(t+1)](\theta_t)+L_\mathbb{KW}(t,\theta_t).
\end{cases}
\end{equation}
By direct calculation, one can rewrite (\ref{Lemma3.2.equa1}) as:
\begin{align}
&\mathcal{J}(t_0,\tau,x_0,i,u_1,u_2^{\mathbb{K}\mathbb{W}})=x_0^TX_{\mathbb{K}\mathbb{W}}^\tau(t_0,i)x_0+\notag\\
&+\sum_{t=t_0}^\tau {\mathbb E}\left[\left(u_1^T(t)+x_{\mathbb{K}\mathbb{W}}^T(t)(M_{\mathbb{K}\mathbb{W}}^\tau)^T(t,\theta_t)
(S_{\mathbb{K}\mathbb{W}}^\tau)^{-1}(t,\theta_t)\right)S_{\mathbb{K}\mathbb{W}}^\tau(t,\theta_t)\star\Big|\theta(t_0)=i\right].
\end{align}
The conclusion follows from Proposition B.3.  iii).\hfill$\blacksquare$\\

\noindent \textbf{Lemma C.3.} Assume that the  assumptions of Lemma C.2. are fulfilled.
Then for any $\tau>0$, $t_0\in{\cal I}(\tau), i\in{\mathfrak N}$ the following hold:
\begin{itemize}
\item [i)]
\begin{equation}\label{e99a}
0\leq X_\tau(t_0,i)\leq \tilde{X}_{\mathbb KW}(t_0,i)
\end{equation}
where $\tilde{\mathbb{X}}_{\mathbb{KW}}(\cdot)$ is the unique bounded and stabilizing solution of the Riccati difference equation (\ref{Ric_closed}).\\
\item [ii)]
\begin{align}\label{e80a}
&R_{11}(t_0-1,i)+\Pi_{311}(t_0-1)[{\mathbb X}_{\tau}(t_0)](i)-(R_{12}(t_0-1,i)+\Pi_{312}(t_0-1)[{\mathbb X}_{\tau}(t_0)](i))\times\notag\\
&\times(R_{22}(t_0-1,i)+\Pi_{322}(t_0-1)[{\mathbb X}_{\tau}(t_0)](i))^{-1}\star \leq -\tilde \nu_1 I_{m_1}
\end{align}
\begin{eqnarray}\label{e80b}
R_{22}(t_0-1,i)+\Pi_{322}(t_0-1)[{\mathbb X}_{\tau}(t_0)](i)\geq \tilde \nu_2 I_{m_2}
\end{eqnarray}
where $\tilde\nu_j>0, j=1,2$ are constant not depending upon $t_0, \tau, i$.
\end{itemize}

\noindent \textbf{Proof.} (i) First, from LemmaC.1.  ii), Lemma C.2. and Proposition B.3. ii), it follows that $X_\tau(t_0,i)\leq \tilde{X}_{\mathbb KW}(t_0,i)$ for every $t_0\in \mathcal{I}(\tau)$, $i\in \mathfrak{N}$.
We have now to prove that $X_\tau(t_0,i)\geq 0$ for every $t_0\in \mathcal{I}(\tau)$, $i\in \mathfrak{N}$. It follows from $X_\tau(t_0,i)\leq \tilde{X}_{\mathbb KW}(t_0,i)$, $(t_0,i)\in \mathcal{I}(\tau)\times \mathfrak{N}$, and (\ref{sign_closed}) that:
\begin{equation}\label{L3.3_equa1}
R_\mathbb{W}(t_0,i)+\Pi_\mathbb{W}(t_0)[{\mathbb{X}}_\tau(t_0)](i)\leq -\xi \mathbb{I}
\end{equation}
for some positive scalar $\xi$, $(t_0,i)\in \mathcal{I}(\tau)\times \mathfrak{N}$. Note also that (\ref{L3.3_equa1}) can be rewritten as:\\
\begin{equation}\label{L3.3_equa2}
\left(\begin{array}{cc}I & W^T(t_0,i)\end{array}\right)\left(R(t_0,i)+\Pi_3(t_0)[\mathbb{X}_\tau(t_0+1)](i)\right)\star\leq -\xi \mathbb{I}.
\end{equation}
Now using Lemma A.1 in Appendix A with:
\begin{equation}\label{L3.3_equa3}
\begin{cases}
\Gamma_1(t_0,i)=0\\
\Gamma_2(t_0,i)=F_2(t_0,i)-W(t_0,i)F_1(t_0,i)
\end{cases}
;\;(t_0,i)\in \mathcal{I}(\tau)\times \mathfrak{N}
\end{equation}
one shows that the Riccati difference equation (\ref{GRDE}) can be rewritten as:
\begin{align}\label{L3.3_equa4}
X(t_0,i)&=\sum_{k=0}^r(A_k(t_0,i)+B_k(t_0,i)\Gamma(t_0,i))^T\Xi(t_0)[\mathbb{X}(t_0+1)](i)\star+\mathcal{M}_{\Gamma}(t_0,i)-\notag\\
&-F_1^T(t_0,i)\left[\begin{array}{cc}I & W^T(t_0,i)\end{array}\right]\left(R(t_0,i)+\Pi_3(t_0)[\mathbb{X}_\tau(t_0+1)](i)\right)\left[\begin{array}{c}I \\W(t_0,i)\end{array}\right]F_1(t_0,i)
\end{align}
for $(t_0,i)\in \mathcal{I}(\tau)\times \mathfrak{N}$. It follows from the parametrization in (\ref{L3.3_equa3}) that:
\begin{align}\label{L3.3_equa5}
\mathcal{M}_{\Gamma}(t_0,i)=M(t_0,i)-L_2(t_0,i)R_{22}^{-1}(t_0,i)\star+\left(\Gamma_2(t_0,i)+R_{22}^{-1}(t_0,i)L_2^T(t_0,i)\right)R_{22}(t_0,i)\star
\end{align}
The conclusion follows by induction using assumption \textbf{H4}), (\ref{L3.3_equa2}) and (\ref{L3.3_equa4}).

(ii) From (\ref{sign1}) and (\ref{e99a}) it follows that
$$R_{22}(t_0-1,i)+\Pi_{322}(t_0-1)[{\mathbb X}_{\tau}(t_0)](i)\geq R_{22}(t_0-1,i)\geq \rho_2 I_{m_2}$$
for all $\tau>0, t_0\in{\cal I}(\tau), t_0\geq 1, i\in{\mathfrak N}$. Thus we have obtained that (\ref{e80b}) is satisfied for $\tilde \nu_2=\rho_2>0$.
On the other hand, from (i) we deduce that
\begin{eqnarray}\label{e80c}
R(t_0-1,i)+\Pi_3(t_0-1)[{\mathbb X}_{\tau}(t_0)](i)\leq R(t_0-1,i)+\Pi_3(t_0-1)[\tilde {\mathbb X}_{\mathbb KW}(t_0)](i).
\end{eqnarray}
Employing Corollary 2.6 from \cite{Freiling:03} in the case of the matrices from (\ref{e80c}) together with Proposition B.1. we obtain that (\ref{e80a}) is satisfied with $\nu_1=\varepsilon^2$ given in (\ref{Eq_3.45}).
\hfill$\blacksquare$


\begin{thebibliography}{99}

\bibitem{abu} H. Abou-Kandil, G. Freiling, V. Ionescu and G. Jank, Matrix Riccati equations in
Control Systems Theory, \textit{Birhauser, Basel}, 2003.

\bibitem{SV:Automatica} 
S. Aberkane and V. Dragan, $H_\infty$ filtering of periodic Markovian jump systems: Application to filtering with communication constraints, {\it Automatica}, 48 (12), 3151--3156, 2012.

\bibitem{samirieee}
S. Aberkane, Bounded real lemma for nonhomogeneous Markovian jump linear systems, {\em IEEE Transactions on Automatic Control}, 58 (3), 797--801, 2013.

\bibitem{SV} S. Aberkane and V. Dragan, Robust stability and robust stabilization of a class of discrete-time time-varying linear stochastic systems, {\em SIAM Journal on Control and Optimization}, 34 (8), 831--847, 2015.

\bibitem{Damm:02} T. Damm, Rational matrix equations in stochastic control, \textit{Lecture notes in control and information sciences, Springer}, 2002.

\bibitem{carte2010} V. Dragan, T. Morozan and A. M. Stoica, Mathematical Methods in Robust Control of Discrete-Time Linear Stochastic Systems, \textit{Springer, New York}, 2010.

\bibitem{aosr2010}
V. Dragan, T. Morozan, Robust stability and robust stabilization of discrete-
time linear stochastic systems, \textit{Annals of the Academy of Romanian Scientists, Series on
Mathematics and its Applications}, 2 (2), 141--170, 2010.

\bibitem{Dragan:14} V. Dragan, Robust stabilization of discrete-time time-varying
linear systems with Markovian switching and nonlinear parametric uncertainties, {\it International Journal of Systems Science}, 45 (7), 1508--1517, 2014.

\bibitem{Dragan:18}  V.~Dragan, S. Aberkane and T. Morozan, On the bounded and stabilizing solution of a generalized Riccati differential equation arising in connection with a zero-sum linear quadratic stochastic differential game, {\em Submitted}, 2018.

\bibitem{Freiling:03} G. Freiling and A. Hochhaus, Properties of the solutions of rational matrix difference equations, {\it Advances in difference equations, IV, Comput. Math. Appl.}, 45, 1137-1154, 2003.

\bibitem{Halanay:93}
A. Halanay, and V. Ionescu, Time-varying discrete linear systems, \textit{Birkhäuser}, 1994.

\bibitem{ma}
H. J. Maa, W. Zhang, and T. Hou, Infinite horizon $H_2$/$H_\infty$ control for discrete-time time-varying Markov jump systems with multiplicative noise, \textit{Automatica}, 48 (7), pp. 1447--1454, 2012.

\bibitem{mou} M.McAsey and L. Mou, Generalized Riccati equations arising in stochastic games, {\it Linear Algebra and its Appl.}, 416 (2006) 710-–723.

\bibitem{morozan1998}
T. Morozan, Parametrized Riccati equations for controlled linear discrete-time
systems  with Markov perturbations, \textit{Revue Roumaine de Math\'ematique Pures et Appliqu\'ees}, 43 (7-8), 761--777, 1998.


\bibitem{ungdrmo}
V. M. Ungureanu, V. Dragan and T. Morozan, Global solutions of a class of discrete-
time backward nonlinear equations on ordered Banach spaces with applications to Riccati
equations of stochastic control, \textit{Optimal Control, Applications and Methods}, 34 (2), 164--
190, 2013.

\bibitem{Yu:2015} Z. Yu, An Optimal Feedback Control-Strategy Pair for Zero-sum Linear-Quadratic Stochastic Differential Game: The Riccati Equation Approach, {\it SIAM journal on control and optimization}, 53 (4), 2015, pp. 2141--2167.

\end{thebibliography}
\end{document}